\documentclass[11pt]{article}
\usepackage{amsmath,amssymb,amsfonts}
\usepackage{latexsym}
\usepackage[pdftex]{graphicx}
\usepackage{algorithm,algpseudocode}
\usepackage{times}
\renewcommand\thefootnote{\arabic{footnote}}

\newcommand{\Order}{\mathrm{O}}

\newcommand{\poly}{\mathrm{poly}}
\newcommand{\defeq}{\stackrel{\mbox{\scriptsize{\normalfont\rmfamily def.}}}{=}}

\renewcommand{\Vec}[1]{\mbox{\boldmath $#1$}}
\newcommand{\symmdiff}{\vartriangle}
\newcommand{\circp}{\operatorname{\circ}}
\newcommand{\EO}{{\bf EO}}

\usepackage{amsthm}
\newtheorem{theorem}{Theorem}[section]
\newtheorem{lemma}[theorem]{Lemma}
\newtheorem{proposition}[theorem]{Proposition}
\newtheorem{corollary}[theorem]{Corollary}
\newtheorem{observation}[theorem]{Observation}

\newtheorem{problem}{Problem}

\newtheorem{Algorithm}{Algorithm}

\title{Finding Submodularity Hidden in Symmetric Difference}
\author{
 Junpei Nakashima\footnotemark[1] \and 
 Yukiko Yamauchi\footnotemark[1] \and
 Shuji Kijima\footnotemark[1]~\footnotemark[2] \and
 Masafumi Yamashita\footnotemark[1]
 }
\begin{document}
\maketitle
\renewcommand\thefootnote{\fnsymbol{footnote}}
\footnotetext[1]{
 Graduate School of Information Science and Electronic Engineering, 
 Kyushu University
}
\footnotetext[2]{
JST PRESTO, 744 Motooka, Nishi-ku, Fukuoka, 819-0395, Japan
}
%\footnotetext[3]{
% Professor Emeritus, Kyushu University
%}

\renewcommand\thefootnote{\arabic{footnote}}

%%%%%%%%%%%%%%%%%
\begin{abstract}
 A set function $f$ on a finite set $V$ is {\em submodular}
  if $f(X) + f(Y) \geq f(X \cup Y) + f(X \cap Y)$ for any pair $X, Y \subseteq V$. 
 The {\em symmetric difference transformation} ({\em SD-transformation}) 
  of $f$ by a {\em canonical set} $S \subseteq V$ is a set function $g$ given by
%  of $f$ by a {\em canonical set} $S \subseteq V$ is given by
 $g(X) = f(X \symmdiff S)$ for $X \subseteq V$,
  where $X \symmdiff S = (X \setminus S) \cup (S \setminus X)$ denotes the symmetric difference 
  between $X$ and~$S$.
%%%%%%%%
 Submodularity and SD-transformations 
   are regarded as the counterparts of 
   convexity and affine transformations in a discrete space, respectively. 
 However, 
  submodularity is not preserved under SD-transformations, 
  in contrast to the fact that convexity is invariant under affine transformations. 
%%%%%%%%%%%%%%%%%%%%%%%
 This paper presents 
   a characterization of SD-transformations preserving submodularity. 
 Then, we are concerned with the problem of discovering a canonical set $S$,
  given the SD-transformation $g$ of a submodular function $f$ by $S$,
  provided that $g(X)$ is given by a function value oracle.
 A submodular function $f$ on $V$ is said to be {\em strict}
  if $f(X) + f(Y) > f(X \cup Y) + f(X \cap Y)$ 
 holds whenever both $X \setminus Y$ and $Y \setminus X$ are nonempty.
 We show that the problem is solved by using $\Order(|V|)$ oracle calls when $f$ is strictly submodular,
  although it requires exponentially many oracle calls in general.

%%%%%%%%%%
\end{abstract}

%%%%%%%%%%%%%%%%%%%%%%%
\section{Introduction}
%%%%%%%%%%%%%%%%%%%%%%%%%%%%%%%%%%%%%%%%%%%%%%
\subsection{Submodular function and convexity}
%%%%%%%%
\paragraph{Submodular function on a finite set.}
 For a set function $f\colon 2^V \to \mathbb{R}$ on a finite set $V$, 
  we define
\begin{eqnarray}
  \Phi_f(X,Y) \defeq f(X) + f(Y) -  f(X \cup Y) - f(X \cap Y) 
\label{def:Phi}
\end{eqnarray}
  for any $X,Y \subseteq V$, for convenience of the arguments of the paper. 
%%%%%%%%%%%%%%%%%%
 A set function $f$ is {\em submodular}
  if $\Phi_f(X,Y) \geq 0$ holds\footnote{ 
 Clearly, 
  the condition $\Phi_f(X,Y) \geq 0$ 
   is equivalent to $f(X) + f(Y) \geq f(X \cup Y) + f(X \cap Y)$, which is often used. 
 }  for any pair $X,Y \in 2^V$. 
 In this paper, we do not assume $f(\emptyset)=0$ for a submodular function $f$, 
    which is often assumed in the literature, 
   but this is not essential to the arguments of the paper. 
%%%%%%%%%%%%%%%%
 A submodular function is {\em strictly submodular} 
  if $\Phi_f(X,Y) > 0$ holds whenever both $X \setminus Y$ and $Y \setminus X$ are nonempty. 
%%%%%%%%%%%%%%%%%%
% Similarly, 
 In contrast, 
  a set function is {\em modular} 
  if $\Phi_f(X,Y) = 0$ holds for any pair $X,Y \in 2^V$. 

%%%%%%%%%%%%%%%%%%%%%%%%
 Submodular function 
   is an important concept particularly in the context of combinatorial optimization, and 
  has many applications in economics, machine learning, etc. 
 It is well-known that 
  minimizing a submodular function given as its function value oracle 
  is solved efficiently, 
  by calling the value oracle (strongly) polynomial times~\cite{Schrijver00,IFF01,IO09,LSC15}. 
 In contrast, 
  maximizing submodular function, e.g., max cut, is NP-hard, 
  and approximation algorithms have been developed e.g.,~\cite{NWF78,FMV07}. 

%%%%%%%%%%%%%%
 A celebrated characterization of a submodular function 
   is described by the Lov\'{a}sz extension (see e.g., \cite{Bach,Fujishige,Murota}). 
 For a set function $f \colon 2^V \to \mathbb{R}$, 
  the {\em Lov\'{a}sz extension} $\widehat{f} \colon \mathbb{R}^V \to \mathbb{R}$ is 
  defined for $\Vec{x} = (x(v)) \in \mathbb{R}^V$ which satisfies $x(v_1) \geq x(v_2) \geq \cdots \geq x(v_{|V|})$
  by 
  $\widehat{f}(\Vec{x}) \defeq 
   \sum_{i=1}^{|V|} x(v_i) (f(\{v_j \mid j \leq i \}) - f(\{v_j \mid j \leq i-1 \}) ) + f(\emptyset)$. 
 Lov\'{a}sz~\cite{Lovasz83} showed that 
  a set function $f$ is submodular if and only if $\widehat{f}$ is convex. 
%%%%%%%%%%%%
 There are many other arguments 
    to regard submodular functions 
   as a discrete analogy of convex functions see e.g., \cite{Lovasz83,Murota}.

%%%%%%%%%%%%%
\paragraph{Convex function in continuous space.}
 A function $f \colon \mathbb{R}^n \to \mathbb{R}$ in a continuous space is {\em convex} 
  if $\lambda f(\Vec{x}) + (1-\lambda)f(\Vec{y}) \geq f(\lambda \Vec{x} + (1-\lambda)\Vec{y})$ 
  holds for any $\Vec{x},\Vec{y} \in \mathbb{R}^n$ and $\lambda \in [0,1]$ (see e.g.,~\cite{Rockafellar70,Murota}). 
 An important property of a convex function (even on a convex set) 
   is that local minimality guarantees the global minimality, 
  and convexity is regarded as a tractable and useful class in the context of optimization. 
%%%%%%%%%%%%
 As another property, 
  convexity is invariant under an {\em affine map}; 
  Let $h \colon \mathbb{R}^n \to \mathbb{R}^n$ be an affine map
    given by $h(\Vec{x}) \defeq A\Vec{x}+\Vec{b}$ with some $A \in \mathbb{R}^{n \times n}$ and $\Vec{b} \in \mathbb{R}^n$ and 
  let $f \colon \mathbb{R}^n \to \mathbb{R}$ be a convex function.
  Then, the composition $g \defeq f \circp h$, i.e., $g(\Vec{x}) = f(A\Vec{x}+\Vec{b})$, is again a convex function. 

\paragraph{Change-of-variables for submodular function.}
%%%%%%%%
 A change-of-variables is a fundamental technique for a function. 
 For instance, it may not be trivial whether a continuous function 
   $f(x,y)=8x^2+10xy+3y^2$ is convex or not. 
 Let $x=s-2t$ and $y=-s+3t$, 
  then we get another function $g(s,t) = f(s-2t,-s+3t)=s^2-t^2$. 
 It is relatively easy to see that $g$ is not convex; 
  that is confirmed by 
   $\tfrac{1}{2}g(0,1)+\tfrac{1}{2}g(0,-1) < g(0,0)$ 
   where $g(0,1)=-1$, $g(0,-1)=-1$ and $g=(0,0)=0$. 
 Since convexity is invariant under an affine map, 
  we see that $f$ is not convex.

 This paper is motivated by a ``change-of-variables'' for submodular functions, 
   as a discrete analogy.
 As the counter pert of affine transformations of convex functions, 
  we will investigate  symmetric difference transformations (SD-transformations) of submodular functions, 
  which we will describe just below. 

%%%%%%%%%%%%%%%%%%%%%%%%%%%%%%%%
\subsection{SD-transformation of a submodular function}
%%%%%%%%%%%%%%%%%
% As a discrete analogy of affine maps, 
%  this paper is concerned with symmetric difference maps over a 0-1 hypercube.  
%%%%%%%%%%%%%%%%%%
% Notice that 
%  any bijective map on $2^V$ {\em preserving the $1$-skeleton of the hypercube} (``topology'')
%  is given by a combination of an ``origin-shift''  and renaming of variables. 
% It is trivial that submodularity is invariant under renaming of variables. 
% Then, we are concerned with ``origin-shifts.'' 

%%%%%%
 Let $\sigma_S \colon 2^V \to 2^V$ denote 
  the {\em symmetric difference map} ({\em SD-map}\/) by a set $S \subseteq V$, 
  which is given by 
\begin{eqnarray}
 \sigma_S(X) \defeq X \symmdiff S
\end{eqnarray}
for any $X \subseteq V$,
 where $X \symmdiff S = (X \setminus Y) \cup (Y \setminus X) $ is the symmetric difference between $X$ and $S$.  
%%%%%%%%%%%%%%%%%%%%%%%%%%%%%%
 For a set function $f \colon 2^V \to \mathbb{R}$ and a set $S \subseteq V$, 
  we say $g = f \circp \sigma_S$ is 
  a {\em symmetric difference transformation} ({\em SD-transformation}\/) of $f$ by $S$, 
 i.e., the SD-transformation is the set function 
  $g \colon 2^V \to \mathbb{R}$ given by $g(X) = f(X \symmdiff S)$ for any $X \subseteq V$. 

%%%%%
 It is not difficult to see that 
  any bijective map on $2^V$ preserving the $1$-skeleton of the hypercube (``topology'')
  is given by a combination of an SD-map (``origin-shift'') and renaming elements of $V$. 
 Obviously, submodularity is invariant under renaming elements of $V$. 
 Thus, 
  the SD-transformation is essential in a change-of-variables for submodular functions. 

%%%%%%%%%%%%%%%%%%%
\begin{figure}[tbp]
\begin{center}
 \includegraphics[width=5cm,pagebox=cropbox,clip]{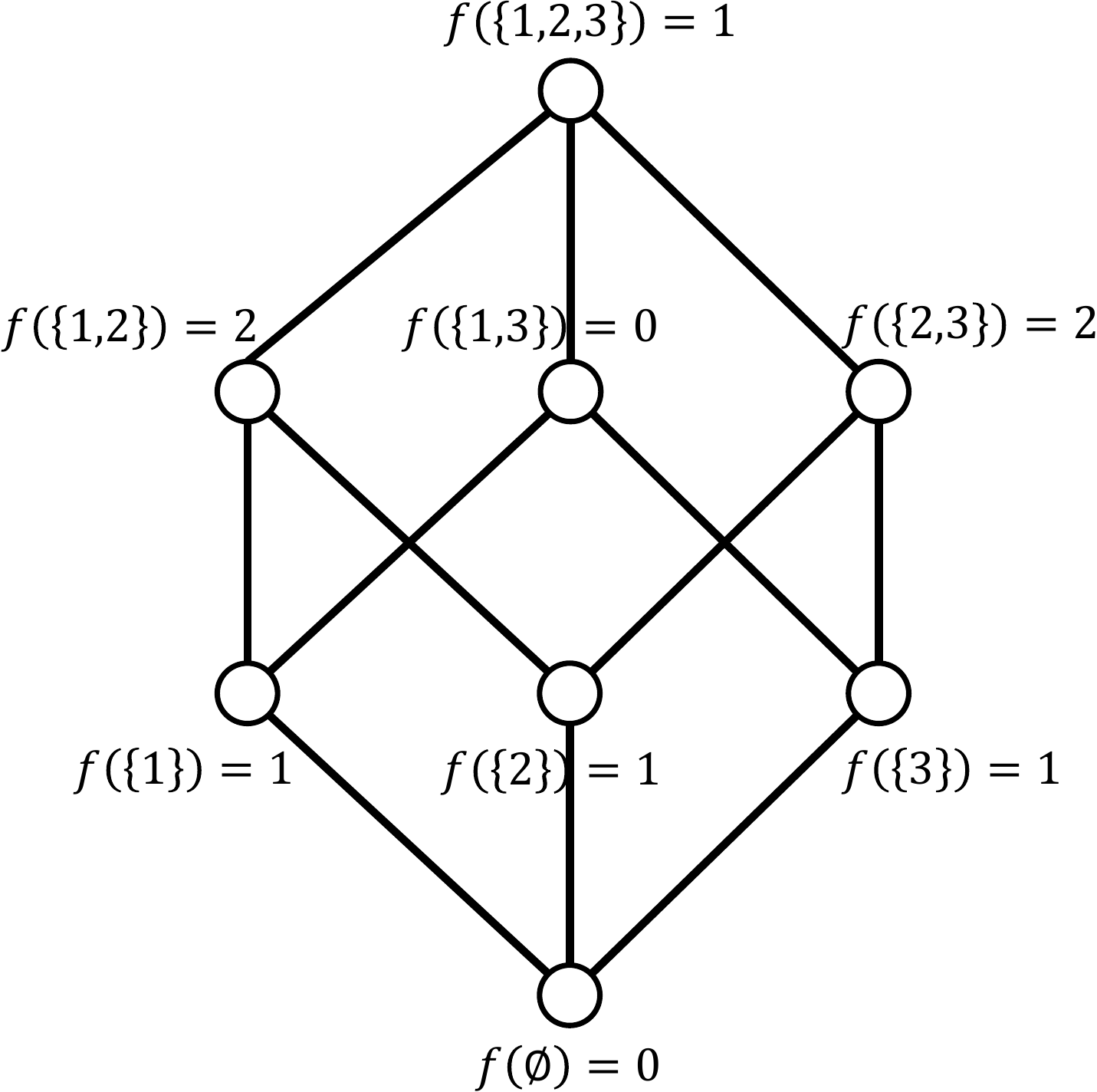}
\hspace{3em}
 \includegraphics[width=5.15cm,pagebox=cropbox,clip]{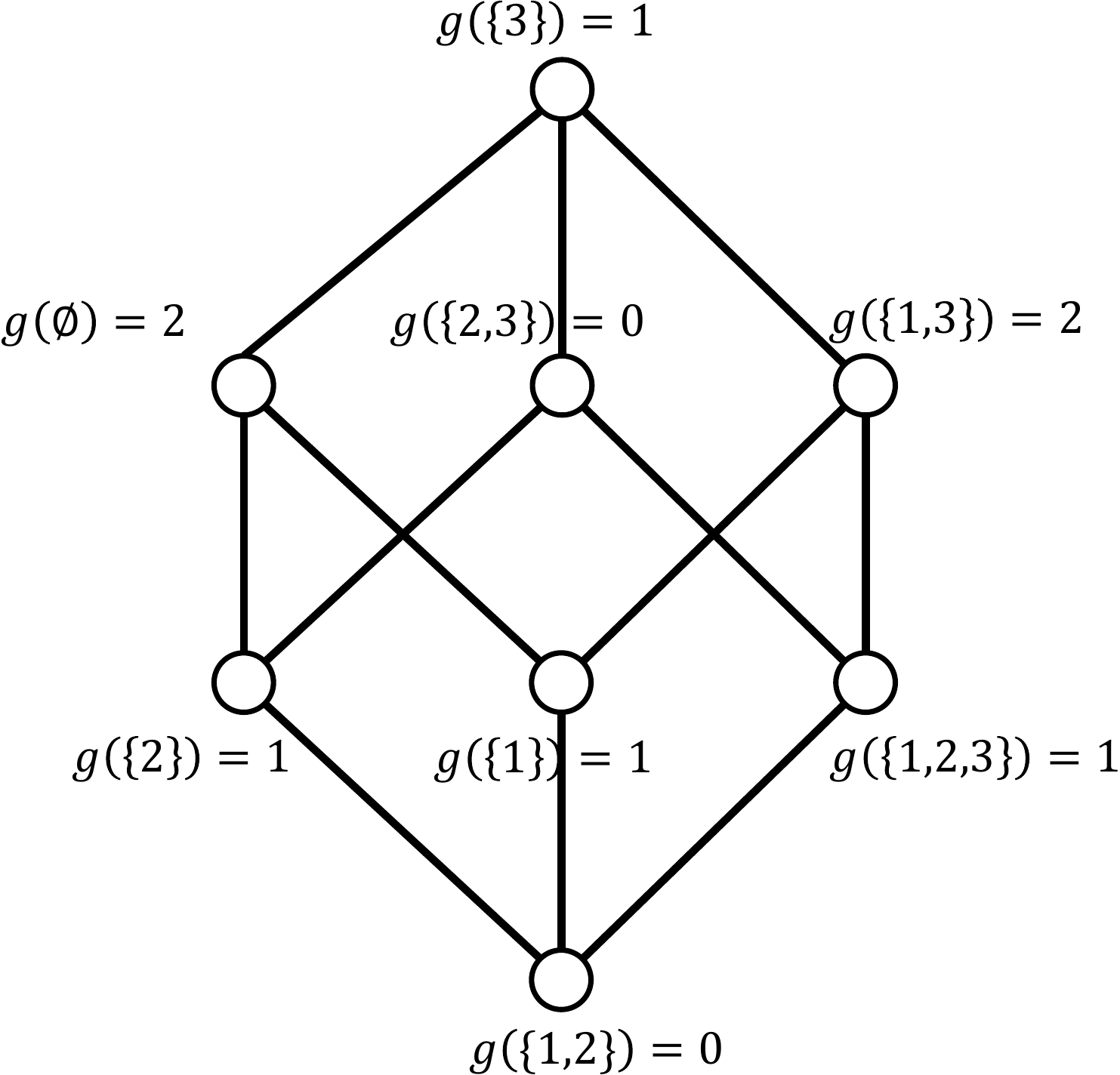}
\end{center}
 \caption{(Left) A submodular function $f$ on $V=\{1,2,3\}$. (Right) $g = f \circ \sigma_{\{1,2\}}$.}\label{fig:1}
\end{figure}
%%%%%%%%%%%%%%
% As we stated above, 
%   submodular functions are often compared with convex functions. 
% It could be natural 
%   to regard SD-maps as discrete analogy of affine maps. 
 However, 
  an SD-transformation of a submodular function is not submodular in general, 
  in contrast to the fact that convexity is invariant under affine maps.  
%%%%%%%%
 Figure~\ref{fig:1} shows an example. 
 The left figure shows a submodular function $f \colon 2^{\{1,2,3\}} \to \mathbb{R}$, and 
  the right figure shows its SD-transformation $g$ by the set $\{1,2\}$. 
 We can check exhaustively that $f$ is submodular, 
  while $g$ is not submodular since $\Phi_g(\{1\},\{3\})=g(\{1\}) + g(\{3\}) - g(\{1,3\}) - g(\emptyset) <0$. 
% The set $\{1,2\}$ is a canonical set of $g$, as well as $\{3\}$, $\{1\}$ and $\{2,3\}$. 

%%%%%%%%%%%%%%%%%%%%%%%%%%%%%%%%%%%%%
\subsection{Contribution}
%%%%%%%%%
% Concerning Problem~\ref{prob:CS}, 
 This paper characterizes SD-maps preserving the submodularity, 
 i.e., 
   given a submodular function $f$, 
   we characterize $S \in 2^V$ for which $f \circp \sigma_S$ is again submodular. 
 In Section~\ref{sec:main}, 
  we present a characterization described by a Boolean system
   (Theorem~\ref{thm:char1}), and 
 rephrase it using a graph defined from~$f$ (Theorem~\ref{thm:char2}).  
%  Theorem~\ref{thm:char2} rephrases the characterization in terms of a graph defined for $f$. 
 By a similar and much simpler argument, we also remark that 
  the modularity is invariant under SD-transformations (Proposition~\ref{prop:modular}). 

% However, 
%  $g \circp \sigma_S$ is submodular again, 
%  since $g \circp \sigma_S = f \circp \sigma_S \circp \sigma_S = f$ holds. 
% This paper is concerned with the following problem.  
Then, we are concerned with the following problem.  
%%%%%%%%%%%%%%%%%%
\begin{problem}\label{prob:CS}
 Let $g\colon 2^V \to \mathbb{R}$ be an SD-transformation of a submodular function. 
 Provided that $g$ is given by its function value oracle, 
  the goal is to find a subset $T \subseteq V$ such that  $h = g \circp \sigma_T$ is submodular. 
\end{problem}
%%%%%%%%
 We call a solution $T$ to Problem~\ref{prob:CS} a {\em canonical set} of $g$. 
 Notice that a canonical set is not unique. 
 In fact, we will show that 
  if $T$ is a canonical set then $V \setminus T$ is also a canonical set 
  (see Proposition~\ref{prop:complement2} in Section~\ref{subsec:complement}). 
%%%%%%%%%
 Once we find a canonical set $T$, 
  we can apply many algorithms for submodular functions, 
  such as minimization or maximization, to $g \circp \sigma_T$. 

%%%%%%%%
 Unfortunately, Problem~\ref{prob:CS} requires exponentially many oracle calls, in the worst case. 
 An easy example is given as follows (see e.g., \cite{GILKB15}).
 Let $U \subseteq V$, then we define a set function $g \colon 2^V \to \mathbb{R}$ by 
  $g(U) = -1$ and $g(X)=0$ for any other subset $X \subseteq V$. 
 Then, the canonical sets are only $U$ and $V \setminus U$. 
 Thus, 
  it is not difficult (intuitively) to see that 
  we need $2^{|V|}-2$ oracle calls in the worst case 
 to solve Problem~\ref{prob:CS} (see also the proof of Proposition~\ref{thm:min-CS}, for a detailed argument). 

%%%%%%%%%%%%%%%
 In Section~\ref{sec:CS}, 
   we present a complete characterization of canonical sets (Theorem~\ref{thm:probCS}). 
   % provides the characterization described by a Boolean system. 
 As an interesting consequence, 
  we show that 
   Problem~\ref{prob:CS} is solved by calling the function value oracle $\Order(|V|)$ times  
	if $f$ is {\em strictly} submodular (Theorem~\ref{thm:strict-submo}).
%%%%%
 Once we find a canonical set of an SD-transformation $g$ of a submodular function, 
  minimization of $g$ is easy using submodular function minimization, as we stated above.  
 However, the converse is not true; 
  we give an example in which 
   Problem~\ref{prob:CS} 
   requires exponentially many function value oracle calls %exponential times 
    even if we have all minimizers (or maximizers) of $g$ (Section~\ref{sec:bad-example}). 

%%%%%%%%%%%%%%%%%%%%%%%%%%%%%
\subsection{Related works}
%%%%%%%%%%%%%%%%%%
\paragraph{Recognizing submodularity.}
 It takes exponential time to check naively 
    if a set function given by its function value oracle is submodular, in general.  
 To be precise, 
  the submodularity is confirmed in $2^{|V|} \cdotp \poly(|V|)$ time, 
  instead of checking $\Phi_f(X,Y) \geq 0$ for all ${2^{|V|} \choose 2} \simeq (2^{|V|})^2$ pairs $X,Y \in 2^{|V|}$  
   (see Section~\ref{sec:preliminary}). 

%%%%%%%%
 Goemans et al.~\cite{GHIM09} is concerned with 
   approximating a submodular function with polynomially many oracle calls. 
 For nonnegative monotone submodular functions $f$, 
  they showed that an approximate function $\widetilde{f}$ is constructed 
   by calling $\poly(|V|)$ times the function value oracle of $f$, 
   such that $\widetilde{f}(X) \leq f(X) \leq \alpha \widetilde{f}(X)$ for any $X \in 2^V$ 
   with an approximation factor $\alpha =\Order(\sqrt{|V|} \log |V|)$. 
 Notice that $\widetilde{f}$ may not be submodular. 
 They also gave a lower bound $\Omega(\sqrt{|V|}/\log |V|)$ 
  of the approximation ratio 
  with polynomially many oracle calls. 
 
%%%%%%%%%%%%%%%%%%%%%%%%%%%%%%%%%%%%%%%%%%%
\paragraph{SD-transformation of a submodular function.} 
 Gillenwater et al.~\cite{GILKB15} are concerned with {\em submodular Hamming distance} 
   $d_f(A,B) = f(A \symmdiff B)$ for $A,B \subseteq V$ 
  given by a positive polymatroid function $f$, 
   that is a monotone nondecreasing positive submodular function $f$ satisfying $f(\emptyset) =0$. 
 Giving some applications in  machine learning, such as clustering, structured prediction, and diverse $k$-best, 
  they investigated the hardness and approximations of problems
   SH-min: $\min_{A \in {\cal C}} \sum_{i=1}^m f_i(A \symmdiff B_i)$ and  
   SH-max: $\max_{A \in {\cal C}} \sum_{i=1}^m f_i(A \symmdiff B_i)$, 
   where $f_i$ is a positive polymatroid, $B_i \subseteq V$, and ${\cal C}$ denotes a combinatorial constraint.

%%%%%%%%%%%%%%%%%%%%%%%%%%%%%%%%%%%%%%%%%%%
\subsection{Organization}
 This paper is organized as follows. 
 As a preliminary step, 
  Section~\ref{sec:preliminary} is concerned with the 2-faces of $0$-$1$ hypercube. 
 More precisely,   
  Section~\ref{sec:2.1} mentions the known fact that 
  the submodularity is confirmed only by checking 
  the submodularity on all 2-faces. 
 Section~\ref{sec:2.2} explicitly writes
  some basic facts of SD-map $\sigma_S$ on 2-faces in concrete terms,  
  to avoid a confusion in the following arguments. 

%%%%%%%%%%
 Section~\ref{sec:main} provides characterizations of SD-maps preserving the submodularity. 
 Prior to the main theorems, 
  Section~\ref{sec:key-lemma} proves a key lemma using the argument in Section~\ref{sec:2.2}. 
 Sections~\ref{sec:char1} and~\ref{sec:char2} 
  respectively show the main theorems. 
 Section~\ref{sec:remark3} make some remarks on Section~\ref{sec:main}. 

%%%%%%%%%%
 Then, 
 Section~\ref{sec:CS} is concerned with Problem~\ref{prob:CS}. 
 Section~\ref{sec:charCS} characterizes canonical sets using a Boolean system. 
 Section~\ref{sec:strict-submo} presents 
   a linear-time algorithm for Problem~\ref{prob:CS} provided that $f$ is strictly submodular. 
 Section~\ref{sec:bad-example} gives some bad examples in general case. 
%%%%%%%%%%
 Finally, Section~\ref{sec:concl} concludes this paper. 

%%%%%%%%%%%%%%%%%%%%%%%
\section{Preliminary: 2-faces of a $0$-$1$ hypercube}\label{sec:preliminary}
%%%%%%%%%%%%%%%%%%%%%%%%%%
\subsection{Submodularity is determined on 2-faces}\label{sec:2.1}
%%%%%%%%%%
 Let $X \in 2^V$, and let $u,v \in V$ be a distinct pair. 
 For convenience, let $X' = X \setminus \{u,v\}$ 
 then
  the four distinct subsets $X'$, $X' \cup \{u\}$, $X' \cup \{v\}$ and $X' \cup \{u,v\}$ of $V$
  form a {\em 2-face} (a.k.a.\ {polygonal face}) of 
   the $n$-dimensional $0$-$1$ hypercube of the vertex set $2^V$. 
%%%%%%
 Let 
\begin{eqnarray}
  {\cal P} \defeq \left\{ (X,\{u,v\}) \;\middle|\; X \subseteq V,  \{u,v\} \in {V \setminus X \choose 2}  \right\}, 
\label{def:calF}
\end{eqnarray}
  representing the whole set of 2-faces of $n$-dimensional hypercube  
  where $(X,\{u,v\})$ corresponds to the 2-face consisting of $X,X\cup \{u\},X\cup \{v\},X\cup \{u,v\}$. 
 Notice that 
 $|{\cal P}| = 2^{n-2}{n\choose 2}$ holds (cf., \cite{Coxeter}). 

%%%%%%%%%%%
 For convenience, 
 let $\check{\Phi}_f \colon {\cal P} \to \mathbb{R}$ be defined by 
\begin{eqnarray}
 \check{\Phi}_f(X,\{u,v\}) 
 &\defeq& \Phi_f(X\cup\{u\},X\cup\{v\}) 
 \label{def:checkPhi}
 \\
 &=& f(X \cup \{u\})+f(X \cup \{u\})-f(X \cup \{u,v\})-f(X)
\nonumber
\end{eqnarray}
 for any $(X,\{u,v\}) \in {\cal P}$.\footnote{
  For the simplicity of the notation, 
   we use the notation $\check{\Phi}(X,\{u,v\})$ instead of $\check{\Phi}((X,\{u,v\}))$. 
  At the same time, 
   we also use the notation $\check{\Phi}(p)$ for $p = (X,\{u,v\}) \in {\cal P}$. 
 }
%%%%%%%%%%%%
 The following characterization of submodular functions is known. 
\begin{theorem}[\cite{Schrijver}]\label{prop:submo-2face}
A set function $f \colon 2^V \to \mathbb{R}$ is submodular if and only if 
 $\check{\Phi}_f(p) \geq 0$ holds for any $p \in {\cal P}$. 
\end{theorem}
%%%%%%%%%%%%%

%%%%%%%%%%%%%%%%%%%%%%%%%%%%%%%%%%%%%%%%%%
\subsection{SD-map on 2-faces}\label{sec:2.2}
 In this section, we are concerned with 
   the map over ${\cal P}$ 
   provided by an SD-map $\sigma_S$ for a subset $S \subseteq V$, 
  as a preliminary step of the arguments in the following sections.
%%%%%
  It is not be difficult to see that 
\begin{eqnarray*}
 \left\{ \sigma_S(X), \sigma_S(X \cup \{u\}), \sigma_S(X \cup \{v\}), \sigma_S(X \cup \{u,v\}) \right\} 
\end{eqnarray*}
  again forms a 2-face of a $0$-$1$ hypercube. 
 To be precise, 
   we can show the following proposition. 
%%%%%%%%%%%%%%%%%%%
\begin{proposition}\label{prop:face}
 For each $(X,\{u,v\}) \in {\cal P}$,
\begin{eqnarray*}
 && \left\{ \sigma_S(X), \sigma_S(X \cup \{u\}), \sigma_S(X \cup \{v\}), \sigma_S(X \cup \{u,v\}) \right\} \\ 
 && = \left\{ Y, Y\cup \{u\}, Y \cup \{v\}, Y \cup \{u,v\} \right\}
\end{eqnarray*}
 holds 
  where $Y = (X \symmdiff S) \setminus \{u,v\}$. 
\end{proposition}
\begin{proof}
 We are concerned with three cases depending on whether $| \{u,v\} \cap S | = 0$, $1$, or $2$. 

%%%%%%%%
Case i)  Suppose that $| \{u,v\} \cap S | = 0$, i.e., $u \not\in S$ and $v \not\in S$. 
 Notice that $u \not\in  X \symmdiff S$ and $v \not\in  X \symmdiff S$. 
 Thus,  
\begin{align}
& Y && &&=X \symmdiff S \label{case:0-1}\\%&&=\sigma_S(X)\\
& Y \cup \{u\} &&= (X \symmdiff S) \cup \{u\}  &&= (X \cup \{u\}) \symmdiff S \label{case:0-2}\\%&&=\sigma_S(X \cup \{u\})\\
& Y \cup \{v\} &&= (X \symmdiff S) \cup \{v\}  &&= (X \cup \{v\}) \symmdiff S \label{case:0-3}\\%&&=\sigma_S(X \cup \{v\})\\
& Y \cup \{u,v\} &&= (X \symmdiff S) \cup \{u,v\} &&= (X \cup \{u,v\} ) \symmdiff S \label{case:0-4}%&&=\sigma_S(X \cup \{u,v\})
\end{align} 
 hold, where the right hand sides are respectively 
  $\sigma_S(X)$, $\sigma_S(X \cup \{u\})$,  $\sigma_S(X \cup \{v\})$ and $\sigma_S(X \cup \{u,v\})$.
 Thus,  we obtain the claim in this case. 

%%%%%%%%%
Case ii) Suppose that $| \{u,v\} \cap S | = 1$. 
 Without loss of generality, 
 we may assume that $u \in S$ and $v \not\in S$.  
 Then, 
%%%%%%%%
\begin{align}
& Y &&= (X \symmdiff S) \setminus \{u\} &&=(X \cup \{u\}) \symmdiff S \label{case:1-1}\\ %&&=\sigma_S(X \cup \{v\})\\
& Y \cup \{u\} && &&= X \symmdiff S \label{case:1-2}\\ %&&=\sigma_S(X)\\
& Y \cup \{v\} &&= ((X \symmdiff S)\setminus\{u\}) \cup \{v\} &&= (X \cup \{u,v\}) \symmdiff S \label{case:1-3}\\ %&&=\sigma_S(X \cup \{u,v\})\\
& Y \cup \{u,v\} &&= (X \symmdiff S) \cup \{v\} &&= (X \cup \{v\}) \symmdiff S \label{case:1-4}%&&=\sigma_S(X \cup \{u\}) 
\end{align} 
 hold, where the right hand sides are respectively 
  $\sigma_S(X \cup \{u\})$, $\sigma_S(X)$,  $\sigma_S(X \cup \{u,v\})$ and $\sigma_S(X \cup \{v\})$.
 Thus,  we obtain the claim in this case. 

%%%%%%%%%%
Case iii) Suppose that $| \{u,v\} \cap S | = 2$, i.e., $u \in S$ and $v \in S$. 
 Then,  
\begin{align}
& Y &&= (X \symmdiff S) \setminus\{u,v\} &&=(X \cup \{u,v\})  \symmdiff S \label{case:2-1}\\ %&&=\sigma_S(X \cup \{u,v\})\\
& Y \cup \{u\} &&= (X \symmdiff S) \setminus\{v\} &&= (X \cup \{v\}) \symmdiff S \label{case:2-2}\\ %&&=\sigma_S(X \cup \{v\})\\
& Y \cup \{v\} &&= (X \symmdiff S) \setminus\{u\} &&= (X \cup \{u\}) \symmdiff S \label{case:2-3}\\ %&&=\sigma_S(X \cup \{u\})\\
& Y \cup \{u,v\} && &&= X \symmdiff S \label{case:2-4}%&&=\sigma_S(X)
\end{align} 
 hold, where the right hand sides are respectively 
  $\sigma_S(X \cup \{u,v\})$, $\sigma_S(X \cup \{v\})$,  $\sigma_S(X \cup \{u\})$ and $\sigma_S(X)$.
 Thus,  we obtain the claim. 
\end{proof}

%%%%%%%%%%%%
 Let 
  $\check{\sigma}_S \colon {\cal P} \to {\cal P}$ for $S \subseteq V$ be defined by 
\begin{eqnarray} 
 \check{\sigma}_S(X,\{u,v\}) \defeq \left(Y,\{u,v\}\right)
 \label{def:sigmaF}
\end{eqnarray}
 for any $(X,\{u,v\}) \in {\cal P}$ where $Y = (X \symmdiff S) \setminus \{u,v\}$.\footnote{
   For the simplicity of the notation we use $\check{\sigma}_S(X,\{u,v\})$, instead of $\check{\sigma}_S((X,\{u,v\}))$. 
   At the same time, 
    we also use the notation $\check{\sigma}(p)$ for $p = (X,\{u,v\}) \in {\cal P}$. 
 }
 Then, $\check{\sigma}_S$ is the map on ${\cal P}$ provided by $\sigma_S$ 
  by Proposition~\ref{prop:face}. 
 Since $\sigma_S$ is bijective on $2^V$, 
  Proposition~\ref{prop:face} also implies the following. 
\begin{corollary}\label{cor:sigmaF-bijective}
 $\check{\sigma}_S$ is bijective. 
\end{corollary}

%%%%%%%%%%%%%%%%%%%%%%%%%%%%%%%
\section{SD-maps preserving submodulraity}
\label{sec:main}
%%%%%%%%%%%%%%%%%%%%%%%%%%%%%%%
 This section characterizes $S \subseteq V$ for which $f \circp \sigma_S$ is submodular. 
 Theorem~\ref{thm:char1} describes it using a Boolean system, and 
 Theorem~\ref{thm:char2} rephrases it using a graph. 
 As a preliminary argument, we give a key lemma in Section~\ref{sec:key-lemma}

%%%%%%%%%%%%%%%%%%%%%%%
\subsection{Key lemma}\label{sec:key-lemma}
%%%%%%%%%%%%
 This section mainly proves a technical lemma (Lemma~\ref{lem:key-lemma}), 
  which intuitively characterizes SD-maps $\sigma_S$ preserving {\em submodularity on a 2-face}. 
 We also remark 
  that any SD-map $\sigma_S$ preserves {\em modularlity on a 2-face} 
  by a similar argument, at Lemma~\ref{lem:key-lemma2}.   
\begin{lemma}\label{lem:key-lemma}
 Let $f \colon 2^V \to \mathbb{R}$ be a submodular function. 
 Suppose that a 2-face $p =(X,\{u,v\}) \in {\cal P}$ satisfies that 
\begin{eqnarray*}
  \check{\Phi}_f(p) > 0. 
\end{eqnarray*}
  Then for any subset $S \subseteq V$, 
\begin{eqnarray}
 \check{\Phi}_{f \circp \sigma_S}(\check{\sigma}_S(p)) >  0 
 \label{eq:key-lemma}
\end{eqnarray}
 holds 
 if and only if  $|S \cap \{u,v\}| \equiv 0 \pmod{2}$. 
\end{lemma}
%%%%%%%%%%%
\begin{proof}
 By Proposition~\ref{prop:face} and the definition \eqref{def:checkPhi} of $\check{\Phi}_{f}$, 
 the condition \eqref{eq:key-lemma} holds if and only if 
\begin{eqnarray*}
 f \circp \sigma_S (Y \cup \{u\}) + f \circp \sigma_S (Y \cup \{v\}) 
   > f \circp \sigma_S(Y \cup \{u,v\}) + f \circp \sigma_S(Y)
\end{eqnarray*}
   holds on the 2-face $(Y,\{u,v\}) = \check{\sigma}(X,\{u,v\})$, 
   where $Y= (X \symmdiff S) \setminus \{u,v\}$. 

($\Leftarrow$)
 We show that \eqref{eq:key-lemma} holds if $|S \cap \{u,v\}|=0$ or 2. 
%%%%%%%%
 If $| \{u,v\} \cap S | = 0$, then 
\begin{eqnarray*}
\lefteqn{\check{\Phi}_{f \circp \sigma_S}(\check\sigma_S(X,\{u,v\}))}\\
&& = f \circp \sigma_S (Y \cup \{u\}) + f \circp \sigma_S (Y \cup \{v\}) 
   - f \circp \sigma_S(Y \cup \{u,v\}) - f \circp \sigma_S(Y) \\
&& =
     f ((Y \cup \{u\}) \symmdiff S) + f ((Y \cup \{v\}) \symmdiff S)
    -f ((Y \cup \{u,v\}) \symmdiff S) - f (Y \symmdiff S) \\
&& =
   f (X \cup \{u\}) + f (X\cup \{v\}) 
 - f (X \cup \{u,v\}) - f (X) 
 \hspace{3em} \mbox{(by \eqref{case:0-1}--\eqref{case:0-4})}\\
&& =
  \check{\Phi}_f (X,\{u,v\}) \\
&& >0
\end{eqnarray*}
  hold, where the last inequality follows from the hypothesis that $\check{\Phi}_f(p) > 0$.

%%%%%%%%%%%%%%%%%%%%
 If $| \{u,v\} \cap S | = 2$, i.e., $u \in S$ and $v \in S$, 
  then 
\begin{eqnarray*}
\lefteqn{\check{\Phi}_{f \circp \sigma_S}(\check\sigma_S(X,\{u,v\}))}\\
&& = f \circp \sigma_S (Y \cup \{u\}) + f \circp \sigma_S (Y \cup \{v\}) 
   - f \circp \sigma_S(Y \cup \{u,v\}) - f \circp \sigma_S(Y) \\
&& =
   f ((Y \cup \{u\}) \symmdiff S) + f ((Y \cup \{v\}) \symmdiff S)
 - f ((Y \cup \{u,v\}) \symmdiff S) - f (Y \symmdiff S) \\
&& =
   f (X \cup \{v\}) + f (X\cup \{u\}) 
 - f (X) - f (X \cup \{u,v\}) 
  \hspace{3em} \mbox{(by \eqref{case:2-1}--\eqref{case:2-4})}\\
&& =
  \check{\Phi}_f (X,\{u,v\}) \\
&& >0
\end{eqnarray*}
  hold, where the last inequality follows from the hypothesis that $\check{\Phi}_f(p) > 0$.
 We obtain the claim. 

%%%%%%%%%%%%%%%
($\Rightarrow$)
 We prove the contrapositive: 
  if $|S \cap \{u,v\}| = 1$ then \eqref{eq:key-lemma} does not hold. 
 Without loss of generality we may assume that $u \in S$ and $v \not\in S$. 
 Then 
\begin{eqnarray*}
\lefteqn{\check{\Phi}_{f \circp \sigma_S}(\check\sigma_S(X,\{u,v\}))}\\
&& = f \circp \sigma_S (Y \cup \{u\}) + f \circp \sigma_S (Y \cup \{v\}) 
   - f \circp \sigma_S(Y \cup \{u,v\}) - f \circp \sigma_S(Y) \\
&& =
   f (Y \cup \{u\}) \symmdiff S) + f ((Y \cup \{v\}) \symmdiff S)
 - f ((Y \cup \{u,v\}) \symmdiff S) - f (Y \symmdiff S) \\
&& =
   f (X) + f (X \cup \{u,v\}) 
 - f (X \cup \{u\}) - f (X\cup \{v\}) 
 \hspace{3em} \mbox{(by \eqref{case:1-1}--\eqref{case:1-4})}\\
&& =
 -\check{\Phi}_f (X,\{u,v\}) \\
&& < 0
\end{eqnarray*}
  hold, where the last inequality follows from the hypothesis that $\check{\Phi}_f(p) > 0$. 
 Now, we obtain the claim. 
\end{proof}

%%%%%%%%%%%%
 The proof of the following lemma is similar to and much easier than the proof of Lemma~\ref{lem:key-lemma}, 
 so it is omitted. 
\begin{lemma}\label{lem:key-lemma2}
 Let $f \colon 2^V \to \mathbb{R}$ be a submodular function. 
 Suppose that a 2-face $p \in {\cal P}$ satisfies that  
\begin{eqnarray*}
 \check{\Phi}_f(p) = 0. 
\end{eqnarray*}
  Then for any subset $S \subseteq V$, 
\begin{eqnarray*}
 \check{\Phi}_{f \circp \sigma_S}(\check{\sigma}_S(p)) = 0
\end{eqnarray*}
 holds. 
\end{lemma}

%%%%%%%%%%%
\subsection{A characterization by a Boolean system}\label{sec:char1}
%%%%%%%%%%%%%%%%%%%%%%%%%%
%%%%%%%%%%%%%%%%%%%%%%%%%%
\begin{figure}[tbp]
\begin{center}
 $M_f = 
  \begin{pmatrix}
  0&0&0 \\
  1&0&1 \\ 
  0&0&0 \\
  0&0&0 \\
  1&0&1 \\ 
  0&0&0
  \end{pmatrix}
  \hspace{0.5em}
  \begin{array}{lr}
  \leftarrow&(\emptyset,\{1,2\})\mbox{-th} \\
  \leftarrow&(\emptyset,\{1,3\})\mbox{-th} \\
  \leftarrow&(\emptyset,\{2,3\})\mbox{-th} \\
  \leftarrow&(\{1\},\{2,3\})\mbox{-th} \\
  \leftarrow&(\{2\},\{1,3\})\mbox{-th} \\
  \leftarrow&(\{3\},\{1,2\})\mbox{-th} 
  \end{array}
  $
\end{center}
 \caption{
   The matrix $M_f$ of the submodular function $f$ given in Figure~\ref{fig:1}.}\label{fig:2.5}
\end{figure}
%%%%%%%%%%%%
% Section~\ref{sec:char1} presents 
 This section presents 
  a characterization of SD-maps preserving submodularity,  
  by using Lemmas~\ref{lem:key-lemma} and \ref{lem:key-lemma2}. 
%%%%%%%%
 For any set function $f \colon 2^V \to \mathbb{R}$, 
  let $M_f \in 2^{{\cal P} \times V}$ be a matrix 
  whose $(X,\{u,v\})$-th row vector is given by 
\begin{eqnarray}
  M_f\left[(X,\{u,v\}), \cdot\right] = \begin{cases}
  \Vec{\chi}_{\{u,v\}}^\top & 
   \text{if $\check{\Phi}_f(X,\{u,v\}) \neq 0$,} \\
  \Vec{0}^{\top} & \text{otherwise,}
  \end{cases}
\label{def:Mf}
\end{eqnarray}
 for each $(X,\{u,v\}) \in {\cal P}$, 
  where $\Vec{\chi}_S \in 2^V$ denotes the characteristic column-vector of $S \subseteq V$, 
  i.e., $\chi_S(w)=1$ if $w \in S$; otherwise $\chi_S(w)=0$.  
 Figure~\ref{fig:2.5} shows the matrix $M_f$ of the submodular function $f$ given in Figure~\ref{fig:1}. 
 Then an SD-map $\sigma_S$ that preserves 
     the submodularity of a submodular function $f$ is characterized by the next theorem. 
%%%%%%%%%%%%%%%%%%%%%%%%%%
\begin{theorem} \label{thm:char1}
 Let $f \colon 2^V \to \mathbb{R}$ be a submodular function. 
 For any $S \subseteq V$,  
  $f \circp \sigma_S$ is submodular 
   if and only if  $M_f  \Vec{\chi}_S \equiv \Vec{0} \pmod{2}$ holds.
\end{theorem}

 To prove Theorem~\ref{thm:char1}, we show the following lemma. 
%%%%%%%%%%%%%%%%%%
\begin{lemma}\label{lem:char1}
 Let $f \colon 2^V \to \mathbb{R}$ be a submodular function. 
 For any $S \subseteq V$ and for any $p=(X,\{u,v\}) \in {\cal P}$, 
 $\check{\Phi}_{f \circp \sigma_S}(\check{\sigma}_S(p)) \geq 0$
  if and only if 
\begin{eqnarray*}
 M_f[p,\cdot] \Vec{\chi}_S \equiv 0\pmod{2}
\end{eqnarray*}
holds. 
\end{lemma}
%%%%%%%
\begin{proof}
 The proof is by a case analysis of $(X,\{u,v\}) \in {\cal P}$. 
 If $\check{\Phi}_f(X,\{u,v\}) = 0$ holds for  $(X,\{u,v\}) \in {\cal P}$, 
  then $M_f[(X,\{u,v\}),\cdot] = \Vec{0}$ holds by the definition of $M_f$. 
 Using Lemma~\ref{lem:key-lemma2}, 
  the claim is easy in this case. 

 Suppose that $\check{\Phi}_f(X,\{u,v\}) \neq 0$  holds for  $(X,\{u,v\}) \in {\cal P}$. 
 Then,  
\begin{eqnarray*}
 M_f[(X,\{u,v\}),\cdot] \Vec{\chi}_S 
 &=& \chi_S(u) + \chi_S(v) \\
  &=& |\{u,v\} \cap S|
\end{eqnarray*}
 by the definition of $M_f$. 
 Now the claim follows from Lemma~\ref{lem:key-lemma}.
\end{proof}

\begin{proof}[Proof of Theorem~\ref{thm:char1}]
 Since $\check{\sigma}_S$ is bijective on ${\cal P}$ by Corollary~\ref{cor:sigmaF-bijective}, 
  $\check{\Phi}_{f \circp \sigma_S}(\check{\sigma}_S(p)) \geq 0$ holds for any $p \in {\cal P}$ 
  implies that  
  $\check{\Phi}_{f \circp \sigma_S}(p) \geq 0$ holds for any $p \in {\cal P}$. 
 Now, 
  Theorem~\ref{thm:char1} is immediate from Lemma~\ref{lem:char1}, using Theorem~\ref{prop:submo-2face}. 
\end{proof}

%%%%%%%%%%%
\subsection{An interpretation of Theorem~\ref{thm:char1} by a graph}\label{sec:char2}
%%%%%%%%%%%%%%%%%%%%%%%%%%
\begin{figure}[tbp]
\begin{center}
\begin{tabular}{c@{\hspace{5em}}c}
\raisebox{3em}{
 $\overline{M}_f = 
  \begin{pmatrix}
  0&0&0 \\
  1&0&1 \\ 
  0&0&0 
  \end{pmatrix}
  \hspace{0.5em}
  \begin{array}{l}
  \leftarrow\{1,2\}\mbox{-th} \\
  \leftarrow\{1,3\}\mbox{-th} \\
  \leftarrow\{2,3\}\mbox{-th} 
  \end{array}
  $
}  &
\includegraphics[width=3cm,pagebox=cropbox,clip]{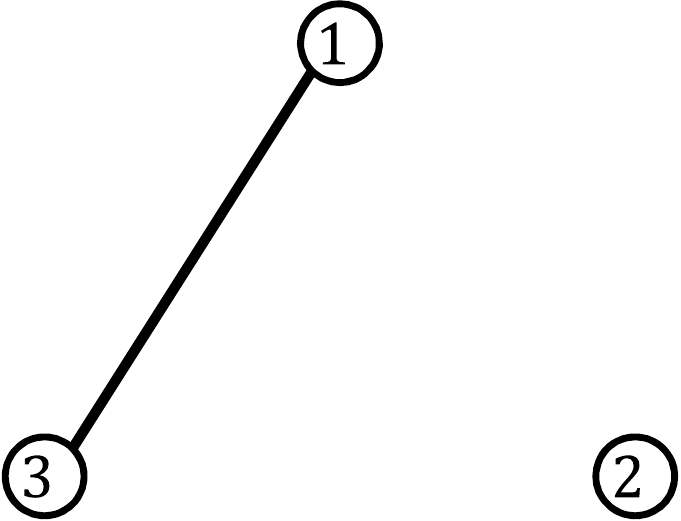}
\end{tabular}
\end{center}
 \caption{
   The reduced matrix $\overline{M}_f$ (left) and 
   the inequality graph $G_f$ (right) 
   of the submodular function $f$ given in Figure~\ref{fig:1}.}\label{fig:2}
\end{figure}
%%%%%%%%%%%%
%%%%%%%
% Section~\ref{sec:char2} interprets the characterization Theorem~\ref{thm:char1} 
 This section interprets the characterization in Theorem~\ref{thm:char1} 
   in terms of a graph defined from $f$. 
%%%%%%%%%
 To begin with, 
  we remark that 
  the Boolean system $M_f \Vec{\chi_S} \equiv \Vec{0} \pmod{2}$ considered in Theorem~\ref{thm:char1} 
  contains many redundant constraints since  
   the rank of the matrix $M_f \in 2^{{\cal P} \times V}$ is at most $|V|$. 
 Then, we define the {\em reduced matrix}\footnote{ 
  The {\em reduced} matrix $\overline{M}_f$ itself does not imply any ``improvement of computational complexity'' to $M_f$: 
  to construct $\overline{M}_f$ we have to check (almost) all $(X,\{u,v\}) \in {\cal P}$ in the worst case, 
  to confirm if $\exists X \subseteq V \setminus \{u,v\}$ satisfying the condition. 
  See also Proposition~\ref{thm:min-CS}. 
} $\overline{M}_f \in 2^{{V \choose 2} \times V}$ of $M_f$ 
  where its $\{u,v\}$-th row vector ($\{u,v\} \in {V \choose 2}$) is given by
\begin{eqnarray}
  \overline{M}_f[\{u,v\}, \cdot] 
 = \begin{cases}
  \Vec{\chi}_{\{u,v\}}^\top & 
   \text{if $\exists(X,\{u,v\}) \in {\cal P}$, 
   $\check{\Phi}_f(X,\{u,v\}) \neq 0$,} \\
  \Vec{0} & \text{otherwise,}
  \end{cases}
\label{eq:redMf}
\end{eqnarray}
 for each $\{u,v\} \in {V \choose 2}$. 
%%%%%%%%%%%%%%%
 % We now make an observation, whose proof is (almost) trivial, so that we omit it.
 We now make an observation, 
   where the ``only-if'' part is trivial, and 
   ``if'' part is not difficult by the definitions of $\overline{M}_f$ (see \eqref{eq:redMf}). 
\begin{observation}\label{obs:1}
 For any $\Vec{\chi} \in 2^V $,
 $M_f \Vec{\chi} \equiv \Vec{0} \pmod{2}$ if and only if $\overline{M}_f \Vec{\chi} \equiv \Vec{0} \pmod{2}$. 
\end{observation}

 We can regard $\overline{M}_f$ as the (redundant) incidence matrix
     of what is called the {\it inequality graph} $G_f$ of $f$.
 Precisely, 
 for any set function $f \colon 2^V \to \mathbb{R}$, 
 let
   $G_f=(V, E_f)$ be an undirected graph with the edge set 
\begin{equation}
  E_f \defeq \left\{ \{u,v\} \in \binom{V}{2} \;\middle|\;
  \mbox{$\exists (X,\{u,v\}) \in {\cal P}$,\ 
   $\check{\Phi}_f(X,\{u,v\}) \neq 0$}
  \right\}. 
\label{def:Ef}
\end{equation}
 Figure~\ref{fig:2} shows the reduced matrix $\overline{M}_f$ and the inequality graph $G_f$ of the submodular function $f$ given in Figure~\ref{fig:1}.

%%%%%%%%%%%%%%%%%%
 The following observation is trivial, too (see also the arguments on the graphic matroid~\cite{Fujishige}).
\begin{observation}\label{obs:2}
 For any $S \subseteq V$, 
  $\overline{M}_f \Vec{\chi}_S \equiv \Vec{0} \pmod{2} $ holds 
  if and only if $\chi_S(u) = \chi_S(v)$ holds for any $\{u,v\} \in E_f$. 
\end{observation}

 Observation~\ref{obs:2} implies that $S$ is a canonical set 
  if and only if every connected component of~$G_f$ is included in or completely excluded from $S$. 
 To be precise, 
  let $U_i \subseteq V$ ($i=1,\ldots,k$) be %a partition of $V$ $indicated by 
  the connected components of $G_f$ %where $k$ is an appropriate integer; 
  where $k$ is the number of connected components of $G_f$.
 Let ${\cal U}(f)$ denote the 
  whole set family of unions of $U_i$ ($i=1,\ldots,k$), i.e., 
\begin{equation}
 {\cal U}(f) = \left\{ \bigcup_{i \in I} U_i \;\middle|\; I \subseteq \{1,2,\dots,k\} \right\}.
 \label{def:calu}
\end{equation}
%%%%%%
 Then, we can conclude the following theorem as an easy
     consequence of Theorem~\ref{thm:char1} and Observations~\ref{obs:1} and~\ref{obs:2}. 
%%%%%%%%%%%%%%%%%
\begin{theorem} \label{thm:char2}
 For any submodular function $f \colon 2^V \to \mathbb{R}$,  
  $f \circp \sigma_S$ is submodular if and only if $S \in {\cal U}(f)$. 
\end{theorem}

%%%%%%%%%%%%%%%%%%%%%%%%%%%%%
\subsection{Remarks on Results in Section~\ref{sec:main}}\label{sec:remark3}
 This section makes some remarks concerning the arguments in Section~\ref{sec:main}. 
 Sections~\ref{subsec:modular} and~\ref{subsec:complement} are easy implications of 
  Lemmas~\ref{lem:key-lemma} and \ref{lem:key-lemma2}. 
 Sections~\ref{sec:nontrivial}, \ref{sec:connected-comp} and \ref{sec:inseparable} 
  are remarks on Theorem~\ref{thm:char2}. 
%%%%%%%%%%%%%%%%
\subsubsection{SD-transformation of a modular function}\label{subsec:modular}
\begin{proposition} \label{prop:modular}
 If a set function $f \colon 2^V \to \mathbb{R}$ is {\em modular} 
  then $f \circp \sigma_S$ is modular for any $S \subseteq V$. 
\end{proposition}
\begin{proof}
 The claim is immediate from Lemma~\ref{lem:key-lemma2}. 
\end{proof}
 We will use Proposition~\ref{prop:modular} in Section~\ref{sec:bad-example}. 

%%%%%%%%%%%%%%
\subsubsection{Complement of a canonical set}\label{subsec:complement}
\begin{proposition} \label{prop:complement}
 If a set function $f \colon 2^V \to \mathbb{R}$ is submodular 
  then $f \circp \sigma_V$ is submodular. 
\end{proposition}
\begin{proof}
 Notice that $M_f \Vec{\chi}_V = \Vec{0}$ holds by the definition \eqref{def:Mf} of $M_f$. 
 The claim is immediate from Theorem~\ref{thm:char1}. 
\end{proof}
 The following proposition for Problem~\ref{prob:CS} immediately follows from Proposition~\ref{prop:complement}. 
\begin{proposition} \label{prop:complement2}
 Let $g \colon 2^V \to \mathbb{R}$ be an SD-transformation of a submodular function. 
 If $T \subseteq V$ is a canonical set of $g$, so is $V \setminus T$. 
\end{proposition}
\begin{proof}
 By the hypothesis, 
  $h =g \circp \sigma_T$ is submodular. 
 Proposition~\ref{prop:complement} implies 
  that $h' = h \circp \sigma_V$ is submodular.  
 Since the symmetric difference is commutative and associative, 
  we see that  
  $h' = h \circp \sigma_V =  g \circp \sigma_T \circp \sigma_V = g \circp \sigma_{V \setminus T}$, and we obtain the claim. 
\end{proof}

%%%%%%%%%%%%%%%%%%%%
\subsubsection{Nontrivial example of many canonical sets}\label{sec:nontrivial}
%%%%%%%%%
 Using Theorem~\ref{thm:char2}, 
  we give a nontrivial example of submodular functions 
  which have many SD-transformations that are submodular. 
%%%%%%%%%%%
 For an arbitrary finite set $V$ and an arbitrary partition $U_1, U_2, \ldots , U_k$ of $V$, 
  let $f \colon 2^V \to \mathbb{R}$
        be a set function defined by
\begin{eqnarray}
  f(X) = \min_{W \in {\cal U}}  |X \symmdiff W| 
\label{ex:part}
\end{eqnarray}
 for any $X \subseteq V$, 
 where ${\cal U} = \left\{ \bigcup_{i \in I} U_i \;\middle|\; I \subseteq \{1,2,\dots,k\} \right\}$. 
%%%%%%%%%% 
 This function represents the edit distance from the nearest point in the Boolean sublattice ${\cal U}$ of $2^V$. 
% (???? any application like genome?)
\begin{proposition}\label{prop:ex-part1}
 The set function $f$ given by \eqref{ex:part} is submodular. 
\end{proposition}
\begin{proposition}\label{prop:ex-part2}
 For the submodular function $f$ given by \eqref{ex:part}, 
 $f \circp \sigma_S$ is submodular if and only if $S \in \mathcal{U}$. 
\end{proposition}

See Appendix~\ref{apx:nontrivial} for the proofs of Propositions~\ref{prop:ex-part1} and~\ref{prop:ex-part2}.  

%%%%%%%%%%%%%%%%%%%%%%%%%%%%%%%%%%
\subsubsection{A connected component of an inequality graph is not a clique, in general}\label{sec:connected-comp}
%%%%%%%%%%%%%%%%%%%%%%%%%%
\begin{figure}[tbp]
\begin{center}
 \includegraphics[width=5cm,pagebox=cropbox,clip]{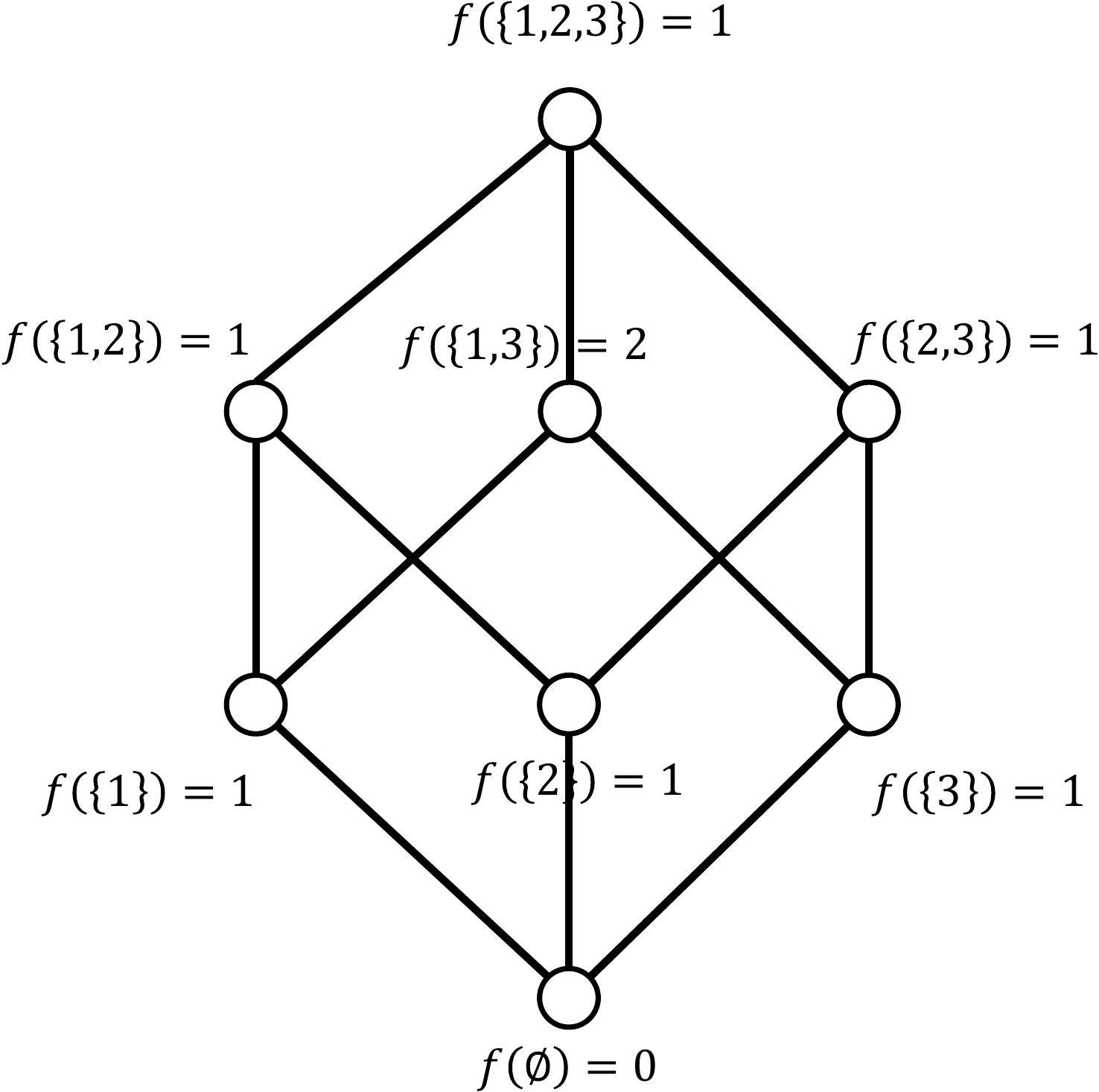}
 \hspace{5em}
 \includegraphics[width=3cm,pagebox=cropbox,clip]{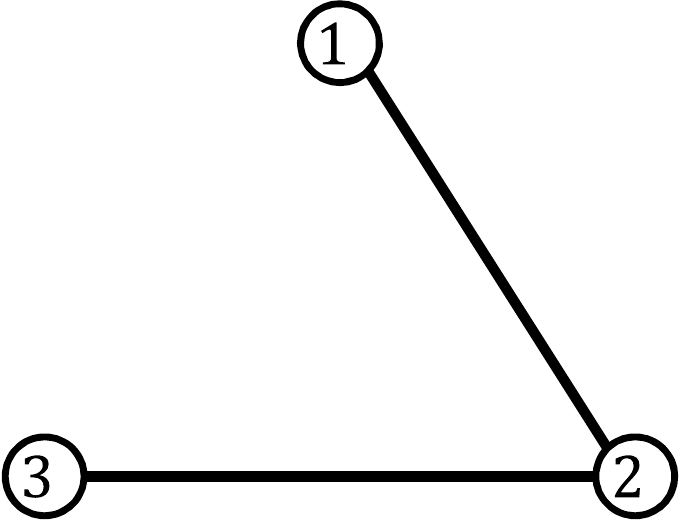}
\end{center}
 \caption{
  (Left) Submodular function $f$ given by \eqref{eq:not_clique} on $V=\{1,2,3\}$. 
  (Right) The inequality graph $G_f$. 
 }\label{fig:not_clique}
\end{figure}
%%%%%%%%%%%%%%%%%%%%%
 As for the inequality graph of a submodular function defined by \eqref{def:Ef}, 
  it would be natural to ask if $U_i$ is a clique, considering the transitivity of $=$. 
 However, it is not true in general. 
 Let $V=\{1,2,3\}$, and 
 let $f \colon 2^V \to \mathbb{R}$ be a set function given by 
\begin{eqnarray}
 f(X) = \begin{cases}
  0& \mbox{if $X=\emptyset$},\\
  2& \mbox{if $X=\{1,3\}$},\\
  1& \mbox{otherwise},
 \end{cases}
\label{eq:not_clique}
\end{eqnarray}
  for $X \in 2^V$ (see Figure~\ref{fig:not_clique} left). 
 We can check that $f$ is submodular by Theorem~\ref{prop:submo-2face}. 
 Then, its inequality graph is given by 
  $G_f = (V,\{\{1,2\},\{2,3\}\})$, 
 since 
   $\check{\Phi}_f(\emptyset, \{1,3\}) = \Phi_f(\{2\}, \{1,3\}) = 0$, 
   $\check{\Phi}_f(\emptyset, \{x,y\}) = 1 \neq 0$, 
   $\check{\Phi}_f(\emptyset, \{y,z\}) = 1 \neq 0$ (see Figure~\ref{fig:not_clique} right). 
 Clearly, the unique connected component of $G_f$ is not a clique.

%%%%%%%%%%%%%%%%%%%%%%%%%%%%%%%%%%%%%%%
\subsubsection{Connection to the inseparable decomposition}\label{sec:inseparable}
%%%%%
 In fact, 
  SD-transformations preserving the submodularity are closely related to  
  the {\em inseparable decomposition} of a submodular function. 
%%%%%%%%%%
 Let $\rho \colon 2^V \to \mathbb{R}$ be 
   a submodular function satisfying $\rho(\emptyset)=0$.\footnote{
    For any submodular function $f$, 
     let $\rho(X) = f(X) - f(\emptyset)$ for $X \in 2^V$, 
     then $\rho$ is submodular satisfying $\rho(\emptyset)=0$.
     } 
 A nonempty subset $U \subseteq V$ is {\em separable} 
   if there exists $X \subset U$ ($X \neq \emptyset$) such that $\rho(U) = \rho(X)+\rho(U \setminus X)$
   holds; otherwise $U$ is {\em inseparable}~\cite{BCT85,Queyranne98,Bach,Kamiyama,Fujishige,Fujishige17}. 
%%%%%%%%%
\begin{theorem}[see e.g., \cite{Kamiyama,Fujishige17}]\label{thm:inseparable}
 For a submodular function $f \colon 2^V \to \mathbb{R}$, 
  $V$ is uniquely partitioned into inseparable subsets $U_1,\ldots,U_k$ with an appropriate $k$. 
 For this partition, 
\begin{eqnarray}
  \rho(X) = \sum_{i=1}^k \rho(X \cap U_i)
\label{eq:171111B1}
\end{eqnarray}
  holds for any $X \in 2^V$.   
 Moreover, this partition is constructible in polynomial time\footnote{
   It takes $\Order(n^2)$ time if we have a base of $f$ (see \cite{Fujishige17,Fujishige}).
 }. 
\end{theorem}

%%%%%%%%%%%%%%%%%%%%
 We can show that $U_1,\ldots,U_k$ form an inseparable decomposition of $\rho$ 
  if and only if each $U_i$ is the vertex set of a connected component of the inequality graph~$G_{\rho}$. 
 See Appendix~\ref{apx:inseparable} for more details. 

%%%%%%%%%%%%%
 Thus, the following theorem is an easy consequence of Theorems~\ref{thm:char2} and \ref{thm:inseparable}. 
\begin{theorem}\label{thm:poly}
 Given a subdmodular function $f \colon 2^V \to \mathbb{R}$ by its function value oracle, and 
  given $S \subseteq V$, 
  the question of $f \circp \sigma_S$ is subdmodular is decidable in polynomial time. 
\end{theorem}
\begin{proof}
 The inseparable decomposition of $f$ can be found in polynomial time by Theorem~\ref{thm:inseparable}, 
  where the decomposition $U_1,\ldots,U_k$ corresponds to the connected components of the inequality graph $G_f$ 
  (see Proposition~\ref{prop:separable}). 
 By Theorem~\ref{thm:char2}, 
   $f \circp \sigma_S$ is subdmodular if and only if $S \in {\cal U}(f)$. 
 The latter condition is checkable in linear time. 
\end{proof}

%%%%%%%%%%%%%%%%%%%%%%%
 We emphasize that 
   Theorem~\ref{thm:poly} does not imply that 
   Problem~\ref{prob:CS}, which is to find {\em unknown} $S$, is solvable in polynomial time. 
 The next section is concerned with Problem~\ref{prob:CS}, 
  using the characterizations given in this section. % Section~\ref{sec:main}.  

%%%%%%%%%%%%%%%%%%%%%%%%%%%%%%%%%%%%%%%%%%%%%%%%%
\section{Finding A Canonical Set}\label{sec:CS}
%%%%%%%%%%%%%%%%%%%%%%%%%%%%%%%%%%%%%%%%%%%%%%%%%%%%%%%%%%%%%%%%%%%%
\subsection{A characterization of canonical sets}\label{sec:charCS}
%%%%%%%%%%
 We are now concerned with Problem~\ref{prob:CS}. 
 Section~\ref{sec:charCS} presents a characterization of   
   canonical sets of an SD-transformation of a submodular function. 
 For any set function $g \colon 2^V \to \mathbb{R}$, 
 let $\Vec{b}_{g} \in 2^{{\cal P}}$ be defined by  
\begin{equation}
  \Vec{b}_{g} [(Z,\{u,v\})] = \begin{cases}
    0 & \mbox{if } \check{\Phi}_g(Z,\{u,v\})\geq 0, \\
    1 & \mbox{otherwise}, 
  \end{cases}
\label{def:bg}
\end{equation}
 for any $(Z,\{u,v\}) \in {\cal P}$. 

\begin{theorem}\label{thm:probCS}
 Let $g \colon 2^V \to \mathbb{R}$ be an SD-transformation of a submodular function. 
 Then, 
  $h=g \circp \sigma_T$ for $T \subseteq V$ is submodular if and only if 
\begin{eqnarray*}
  M_g \Vec{\chi}_T \equiv \Vec{b}_{g} \pmod{2}
\end{eqnarray*}
holds where $\Vec{\chi}_T$ is the characteristic vector of $T$. 
\end{theorem}
\begin{proof}
 Suppose that $g$ is given by $g = f \circp \sigma_S$ for a submodular function $f$ and $S \subseteq V$. 
 Firstly, we claim that 
\begin{eqnarray}
  M_g \Vec{\chi}_S \equiv \Vec{b}_{g} \pmod{2}
\label{eq:20171019}
\end{eqnarray}
 holds. 
  By Lemma~\ref{lem:char1}, 
  $\check{\Phi}_g(\check{\sigma}_S(X,\{u,v\}))\geq 0$ 
   if and only if 
  $  M_f [(X,\{u,v\}),\cdot]  \Vec{\chi}_S\equiv 0 \pmod{2}$.  
%%%%%%%%
 By the definition \eqref{def:Mf} of $M_f$, and 
 Lemmas~\ref{lem:key-lemma} and~\ref{lem:key-lemma2},
\begin{eqnarray}
  M_g[\check{\sigma}_S(X,\{u,v\}), \cdot] =   M_f[(X,\{u,v\}),\cdot]
\label{eq:20171019b}
\end{eqnarray}
  holds for any $(X,\{u,v\}) \in {\cal P}$. 
%%%%%%%%%%
 Thus, 
  $\check{\Phi}_g(\check{\sigma}_S(X,\{u,v\}))\geq 0$ 
   if and only if 
  $    M_g[\check{\sigma}_S(X,\{u,v\}), \cdot] \Vec{\chi}_S \equiv 0 \pmod{2}$. 
 This implies~\eqref{eq:20171019}, 
  since $\check{\sigma}_S$ is bijective on ${\cal P}$ by Corollary~\ref{cor:sigmaF-bijective}. 

%%%%%%%%
 Then, 
 \eqref{eq:20171019} and the hypothesis that  $ M_g \Vec{\chi}_T  = \Vec{b}_g$ imply that 
\begin{eqnarray}
  M_g \Vec{\chi}_S + M_g \Vec{\chi}_T \equiv \Vec{b}_g+\Vec{b}_g \equiv \Vec{0} \pmod{2} 
\label{eq:171105a}
\end{eqnarray}
 holds. Meanwhile, 
\begin{eqnarray}
  M_g \Vec{\chi}_T + M_g \Vec{\chi}_S = M_g (\Vec{\chi}_T + \Vec{\chi}_S) 
\label{eq:171105b}
\end{eqnarray}
 holds. 
 Notice that $\Vec{\chi}_{S \symmdiff T} \equiv \Vec{\chi}_S + \Vec{\chi}_T \pmod{2}$ holds. 
Thus, \eqref{eq:171105a} and \eqref{eq:171105b} imply that 
\begin{eqnarray}
 M_g \Vec{\chi}_{S \symmdiff T} \equiv \Vec{0} \pmod{2}
\label{eq:20171019c}
\end{eqnarray} 
holds. 
By \eqref{eq:20171019b}, 
 \eqref{eq:20171019c} also implies 
\begin{eqnarray}
 M_f \Vec{\chi}_{S \symmdiff T} \equiv \Vec{0} \pmod{2}
\label{eq:20171019d}
\end{eqnarray} 
 holds. 
 Then,  $f \circp \sigma_{S \symmdiff T}$ is submodular by Theorem~\ref{thm:char1}. 
 It is easy to observe that 
\begin{eqnarray*}
g \circp \sigma_T  
= f \circp \sigma_S \circp \sigma_T 
= f \circp \sigma_{S \symmdiff T}
\end{eqnarray*}
 holds, meaning that $T$ is a canonical set. We obtain the claim. 
\end{proof}

%%%%%%%%%%%%%%%%%%%%%%%%%%%%%%%%%%%%%%%%%%%%%%%%%%%%%%%%%%%
\subsection{Linear-time algorithm for strictly submodular function}\label{sec:strict-submo}
%%%%%%%%%%%%%%
 Interestingly, Problem 1 is solvable 
      in linear time for strictly submodular function. 
 Precisely, it is described as follows. 
\begin{theorem}\label{thm:strict-submo}
  Problem~\ref{prob:CS} is solved in $2|V|\cdotp \EO+\Order(|V|)$ time 
  if the set function $f$ is {\em strictly} submodular, 
 where $\EO$ denotes the time complexity of an oracle call to know the value of $g(X)$ for a set $X \subseteq V$. 
\end{theorem}

\begin{proof}
 Since $f$ is strictly submodular, $G_f$ is connected. 
 In particular, 
  let  $u^* \in V$ be arbitrary. 
 Then $\check{\Phi}_g(\emptyset,\{u^*,v\}) \neq 0$ holds for any $v \in V \setminus \{u^*\}$. 
 Thus, we can obtain a canonical set $T \subseteq V$ 
  by solving the Boolean system 
   $M_g[(\emptyset,\{u^*,v\}),\cdot] \Vec{\chi}_T \equiv \Vec{b}_g[(\emptyset,\{u^*,v\})] \pmod{2}$ 
  for $v \in V \setminus \{u^*\}$. 
  (Recall \eqref{def:Mf} and \eqref{def:bg} for the definitions of $M_g$ and $\Vec{b}_g$.) 
 It is not difficult to see that the solution of the Boolean system is 
  a solution of Problem~\ref{prob:CS} 
 by Theorems~\ref{thm:probCS} and \ref{thm:char2}. 

%%%%%%%%%%%%%%%
 In fact, 
  the solution of the Boolean system is simply given as follows: 
 Set $T := \emptyset$ for initialization. 
 For each $v \in V \setminus \{u^*\}$, 
   set $T:=T \cup \{v\}$ if $\check{\Phi}_g(\emptyset,\{u^*,v\}) < 0$. 
 See Algorithm~\ref{alg:strict} for a formal description.
 It is not difficult to observe that 
   the obtained $T$ provides a solution of the Boolean system by Observation~\ref{obs:2}. 
 Computing $\check{\Phi}_g(\emptyset,\{u^*,v\})$ 
  requires the values of 
   $g(\emptyset)$, $g(\{u^*\})$, $g(\{v\})$ and $g(\{u^*,v\})$ 
  for $v \in V \setminus \{u^*\}$. 
 Now the time complexity is easy. 
\end{proof}

\begin{Algorithm}\label{alg:strict}{\normalfont \ 
%\caption{Calculate $y = x^n$}
%\label{alg1}
%\begin{algorithmic}
%Algorithm %[Linear time algorithm for an SD transformation of strictly submodular function]

\qquad  Given a function value oracle of $g \colon 2^V \to \mathbb{R}$. 

\qquad  Set $T := \emptyset$. Choose $u^* \in V$ arbitrarily. 

\qquad  Get the values $g(\emptyset)$ and $g(\{u^*\})$. 

\qquad  For each $v \in V \setminus \{u^*\}$,  

\qquad  \qquad Get the values $g(\{v\})$ and $g(\{u^*,v\})$. 
  
\qquad  \qquad If $\check{\Phi}_g(\emptyset,\{u^*,v\}) < 0$, then set $T:=T \cup \{v\}$. 

\qquad  Output $T$. 
%\end{algorithmic}
}\end{Algorithm}

%%%%%%%%%%%%%%%
\subsection{Minimizer/Maximize is helpless for finding a canonical set}\label{sec:bad-example}
%%%%%%%%%%%%%%
 Once we obtain a canonical set $T$ for an SD-transformation $g$ of a submodular function, 
  we can find the minimum value of $g$ using a submodular function
     minimization algorithm.
 However the opposite is not always true; 
  finding a canonical set is sometimes hard even if all minimizers of $g$ are given. 
\begin{proposition} \label{thm:min-CS}
 Problem~\ref{prob:CS} requires $2^{|V|} - 2$ function value oracle calls in the worst case,
        even if all minimizers of $g$ are given.
\end{proposition}
\begin{proof}
 We give an instance of Problem~\ref{prob:CS}  with a unique minimizer, 
  for which any algorithm needs to call the function value oracle at least $2^{|V|} - 2$ times to solve Problem~\ref{prob:CS}.
 For any $U \subseteq V$ such that $U \neq \emptyset$, 
  let $g_U \colon 2^V \to \mathbb{R}$ be a set function defined by 
\begin{eqnarray}
  g_U (X) = \begin{cases}
    |X|-\dfrac{1}{2}  & (\mbox{if } X = U), \\
    |X| & (\mbox{otherwise}), 
  \end{cases}
\label{ex:min}
\end{eqnarray}
 for $X \in 2^V$. 
 Observe that $g_U(X) > 0$ for any $X \neq \emptyset$,
  meaning that $\emptyset$ is the unique minimizer of $g_U$ with the minimum value $g_U(\emptyset) = 0$. 
 We claim that exactly $U$ and $V \setminus U$ are the canonical sets of $g_U$. 
%%%%%
 Let 
\begin{eqnarray}
 r_U (X) &=& 
\begin{cases}
  -\dfrac{1}{2}  & (\mbox{if } X = U), \\
  0 & (\mbox{otherwise}) 
\end{cases} 
\end{eqnarray}
 for $X \in 2^V$, and let $d(X) \defeq |X|$ for $X \in 2^V$. 
 Then,  
\begin{eqnarray*}
g_U(X) = r_U(X) + d(X)
\end{eqnarray*}
 holds. 
 Clearly $r_U \circp \sigma_U$ is submodular. 
 Since $d$ is a modular function, $d \circp \sigma_U$ is again modular by Proposition~\ref{prop:modular}. 
 Notice that 
\begin{eqnarray*}
 g_U \circp \sigma_U = (r_U + d) \circp \sigma_U =   r_U \circp \sigma_U + d \circp \sigma_U
\end{eqnarray*}
 holds. 
 Since the sum of submodular functions is submodular~\cite{Fujishige}, 
 $g_U \circp \sigma_U$ is submodular, meaning that $U$ is a canonical set of $g$. 
%%%%%%%%%%
 It is easy to observe $G_{r_U}$ is connected, 
   and hence $G_{g_U} = G_{f}$ is connected 
   since $d$ is modular. 
 By Theorem~\ref{thm:char2},  
  we see that only $U$ and $V \setminus U$ are canonical sets of~$g$. 

%%%%%%%%%%%%
 To prove that 
   no algorithm finds a canonical set of $g_U$ with at most $2^{|V|} - 3$ function value oracle calls,  
  we show the existence of an adversarial oracle. 
 Suppose that an arbitrary algorithm calls the value oracle of $g_U$ $2^{|V|}-3$ times. 
 For the $2^{|V|}-3$ queries, our adversarial oracle answers their cardinalities. 
%%%%%%%%%%
 Let $X,Y,Z \in 2^V$ be the remaining sets. 
 Without loss of generality, 
  we may assume that both $X \neq V \setminus Y$ and $X \neq V \setminus Z$ hold. 
 (Note that $Z= V\setminus Y$ may hold.) 
 Since only $U$ and $V \setminus U$ are the canonical sets of $g_U$, 
  both $X$ and $Y$ cannot be canonical sets at the same time. 
 This implies that the algorithm cannot determine $X$, $Y$ or $Z$; 
  if the algorithm answers $X$ then our oracle can set $Y=U$, meaning that $X$ is a wrong answer, and  
  if the algorithm answers $Y$ or $Z$ then our oracle can set $X=U$. 
\end{proof}

%%%%%%%%%%%%%%%
 In contrast to minimization, 
   maximization of a submodular function, e.g., max cut, is NP-hard. 
 Even if all maximizers are given, finding a canonical set is hard, too. 
 The SD-transformation $g_U$ given in the proof of Proposition~\ref{thm:min-CS} 
  also witnesses it. 
\begin{corollary} \label{thm:max-CS}
 Problem~\ref{prob:CS} requires $2^{|V|} - 2$ function value oracle calls in the worst case,
        even if all maximizer of $g$ are given.
\end{corollary}
\begin{proof}
 Let $U \subseteq V$ satisfy $U \neq V$ and 
 let $g_U$ be a set function defined by \eqref{ex:min}. 
 Clearly, $V$ is the unique maximizer of $g$ with the maximum value $g(V)=|V|$. 
 Finding canonical set of $g$ requires $2^{|V|}-2$ oracle calls,
  by the same argument as Proposition~\ref{thm:min-CS}. 
\end{proof}

%%%%%%%%%%%
\subsection{A Remark on Theorem~\ref{thm:probCS}}
 In fact, 
  the hypothesis 
  ``Let $g \colon 2^V \to \mathbb{R}$ be an SD-transformation of a submodular function.''  
  in Theorem~\ref{thm:probCS-2} is redundant. 
 An enhanced theorem is described as follows. 
\begin{theorem}\label{thm:probCS-2}
 Let $g \colon 2^V \to \mathbb{R}$ be an {\em arbitrary set function}. 
 Then, 
  $h=g \circp \sigma_T$ for $T \subseteq V$ is submodular if and only if 
\begin{eqnarray*}
  M_g \Vec{\chi}_T \equiv \Vec{b}_{g} \pmod{2}
\end{eqnarray*}
holds where $\Vec{\chi}_T$ is the characteristic vector of $T$. 
\end{theorem}
 Theorem~\ref{thm:probCS-2} states that 
   if $M_g \Vec{\chi} \equiv \Vec{b}_{g} \pmod{2}$ does not have a solution then 
   any SD-transformation of $g$ is NOT submodular. 
 The ``only-if'' part is immediate from Theorem~\ref{thm:probCS}. 
 The ``if'' part is also not difficult by Theorem~\ref{thm:char1} (and Theorem~\ref{prop:submo-2face}), 
  and obtained in a similar way as the proof of Theorem~\ref{thm:probCS}. 
 We here omit the detailed proof. 

%%%%%%%%%%%%%%%%%%%%%%%%%%%%%%%%%%%%%%%%%%%%%%%
\section{Concluding Remark}\label{sec:concl}
%%%%%%%%%%%%%
 This paper has been concerned with SD-transformations of submodular functions. 
 We gave characterizations of SD-transformations 
    preserving the submodularity in Section~\ref{sec:main}. 
 We also showed that canonical sets are found in linear time 
   for SD-transformations of a {\em strictly} submodular functions in Section~\ref{sec:CS}. 
 It is a natural question 
   whether there is another interesting class of submodular functions for which a canonical set is found efficiently. 
 A related question is 
   whether there is a nontrivial class of transformations (maps) preserving submodularity.

 We remark that it is not difficult 
   to extend the results to submodular functions on {\em distributive lattices}, 
   instead of Boolean lattices. 
 Extensions to submodular functions on a general lattice, i.e., containing $M_3$ or $N_5$, 
  $L$-convex functions and $M$-convex functions~\cite{Murota} on integer lattice,  
  or $k$-submodular functions are interesting problems. 

%%%%%%%%%%%%%%%%%%%%%%%%%%%%
\section*{Acknowledgments}
The authors would like to thank 
  Naoyuki Kamiyama for the note on the inseparable decomposition,  
  and 
  an anonymous reviewer for the advice to mention to Theorem~\ref{thm:probCS-2}. 
The authors are also grateful to 
  Satoru Fujishige, Yusuke Kobayashi and anonymous reviewers 
 for their valuable comments. 
This work is partly supported by JST PRESTO Grant Number JPMJPR16E4, Japan. 

\bibliographystyle{siamplain}
\bibliography{references}

\begin{thebibliography}{99}
\bibitem{Bach}
 F.\ Bach, 
 Learning with Submodular Functions: A Convex Optimization Perspective, 
 Foundations and Trends in Machine Learning, 
 6:2--3, 145--373  (2013). 

\bibitem{BCT85}
 R.\ E.\ Bixby, W.\ H.\ Cunningham and D.\ M.\ Topkis, 
 The partial order of a polymatroid extreme point, 
 Mathematics of Operations Research, 10:3, 367--378 (1985).  

\bibitem{Coxeter}
 H.\ S.\ M.\ Coxeter, 
 Regular Polytopes, Dover, 1973.  

\bibitem{FMV07}
 U.~Feige, V.~Mirrokni and J.~Vondr\'{a}k, 
 Maximizing non-monotone submodular functions, 
 SIAM Journal on Computing, 40, 1133--1153, 2011. 
% Proc.\ the 48th Annual IEEE Symposium on Foundations of Computer Science (FOCS 2007), 
% 461--471. 

\bibitem{Fujishige}
 S.\ Fujishige, 
 Submodular Functions and Optimization, Second Edition, 
 Elsevier, 2005. 

\bibitem{Fujishige17}
 S.\ Fujishige, 
 A note on submodular function minimization by Chubanov's LP algorithm, 
 Optimization online, 6217, 2017.

\bibitem{GILKB15}
 J.\ A.\ Gillenwater, R.\ K.\ Iyer, B.\ Lusch, R.\ Kidambi and J.\ A.\ Bilmes, 
 Submodular Hamming metrics, 
 Proc.\ the 28th International Conference on Neural Information Processing Systems (NIPS 2015), 
 3141--3149. 

\bibitem{GHIM09}
 M.\ X.\ Goemans, N.\ J.\ A.\ Harvey, S.\ Iwata and V.\ S.\ Mirrokni, 
 Approximating submodular functions everywhere, 
 Proc.\ the 20th Annual ACM-SIAM Symposium on Discrete Algorithms (SODA 2009), 
 535--544. 

\bibitem{IFF01}
 S.~Iwata, L. Fleischer and S. Fujishige, 
 A combinatorial, strongly polynomial-time algorithm for minimizing submodular functions, 
 Journal of the ACM, 48, 761--777, 2001. 

\bibitem{IO09}
 S. Iwata and J. Orlin, 
 A simple combinatorial algorithm for submodular function minimization, 
 Proc.\ the 20th Annual ACM-SIAM Symposium on Discrete Algorithms (SODA 2009), 
 1230--1237. 

\bibitem{Kamiyama}
 N.\ Kamiyama, 
 A note on submodular function minimization with covering type linear constraints, 
 Algorithmica, 80, 2957--2971, 2018. 
 
\bibitem{LSC15}
 Y.T.~Lee, A.~Sidford and S.C.~Wong, 
 A faster cutting plane method and its implications for combinatorial and convex optimization, 
 Proc.\ the 56th Annual Symposium on Foundations of Computer Science (FOCS 2015), 
 1049--1065

\bibitem{Lovasz83}
 L.~Lov\'{a}sz, 
 Submodular functions and convexity, 
 in A.~Bachem, B.~Korte B., M.~Gr\"{o}tschel (eds), 
  Mathematical Programming: The State of the Art, 
%  Springer, Berlin, Heidelberg 
 235--257, 1983. 

\bibitem{Murota}
 K.\ Murota, 
 Discrete Convex Analysis, 
 SIAM, 2003. 

\bibitem{NWF78}
 G.\ L.\ Nemhauser, L.\ A.\ Wolsey and M.\ L.\ Fisher, 
 An analysis of approximations for maximizing submodular set functions I, 
 Mathematical Programming, 14 (1978), 265--294. 

\bibitem{Queyranne98}
 M.\ Queyranne,
 Minimizing symmetric submodular functions, 
 Mathematical  Programming, 82, 3--12 (1998). 

\bibitem{Rockafellar70}
 R.\ T.\ Rockafellar, 
 Convex analysis, 
 Princeton University Press, 1970.


\bibitem{Schrijver00}
 A. Schrijver, A combinatorial algorithm minimizing submodular functions in strongly polynomial time, 
 Journal of Combinatorial Theory, Series B, 80, 346--355 (2000). 

\bibitem{Schrijver}
 A.\ Schrijver, 
 Combinatorial Optimization, 
 Springer, 2003. 
\end{thebibliography}

%%%%%%%%%%%%%%%%%%%%%%%%%%%%%%
\appendix
\section{Supplemental Proofs in Section~\ref{sec:nontrivial}}\label{apx:nontrivial}
 This section proves Propositions~\ref{prop:ex-part1} and \ref{prop:ex-part2}.
 The set function which we are concerned with here is given by 
\begin{eqnarray*}
  f(X) = \min_{W \in {\cal U}}  |X \symmdiff W| 
 \hspace{3em}\mbox{(recall \eqref{ex:part})}  
\end{eqnarray*}
 for $X \in 2^V$, where 
  ${\cal U} = \left\{ \bigcup_{i \in I} U_i \;\middle|\; I \subseteq \{1,2,\dots,k\} \right\}$ 
  for a partition $U_1,\ldots,U_k$ of $V$. 

%%%%%%%%%%%%%%%%%%%%%%%%%%%%%%%%%% 
\subsection{Proof of Proposition~\ref{prop:ex-part1}}
%%%%%%%%%%%%%%%%%%%%%5
\begin{proposition}[Proposition~\ref{prop:ex-part1}]
 The set function $f$ given by \eqref{ex:part} is submodular. 
\end{proposition}
%%%%%%%%%%%%%
\begin{proof}
 Since $U_1,\ldots,U_k$ is a partition of $V$, 
\begin{align}
 |X \symmdiff W| 
 &= \sum_{i=1}^m |(X \symmdiff W) \cap U_i| \nonumber\\
 &= \sum_{i=1}^m |(X \cap U_i) \symmdiff (W \cap U_i)|
\label{eq:171107a}
\end{align}
  holds for any $X \in 2^V$ and $W \in {\cal U}$. 
 Notice that 
\begin{align}
 |(X \cap U_i) \symmdiff (W \cap U_i)|
 &= \begin{cases}
 |U_i \setminus X| & \mbox{if $U_i \subseteq W$,} \\
 |X \cap U_i| &\mbox{otherwise, (i.e., $U_i \cap W = \emptyset$ since $W \in {\cal U}$,)}
 \end{cases}
\label{eq:171107b}
\end{align}
 hold. 
%%%%%%%%%%%%%
 Then, 
\begin{align*}
  f (X) 
  &= \min_{W \in {\cal U}} |X \symmdiff W| \\
  &= \min_{W \in {\cal U}} \sum_{i=1}^m |(X \cap U_i) \symmdiff (W \cap U_i)| &&(\mbox{by \eqref{eq:171107a}}) \\
  &= \sum_{i=1}^m \min \{ |U_i \setminus X|,|X \cap U_i|\} &&(\mbox{by \eqref{eq:171107b}})
\end{align*}
 holds for any $X \in 2^V$. 
 For convenience, we define $h_U \colon 2^V \to \mathbb{R}$ for $U \subseteq V$ by 
\begin{eqnarray}
 h_U(X) \defeq \min\{ |X \cap U|, |U \setminus X| \}
\label{def:phi}
\end{eqnarray}
 for $X \in 2^V$. 
 Then, 
  $f(X) = \sum_{i=1}^m h_{U_i}(X)$ 
 holds for any $X \in 2^V$.  
%%%%%%%%%%%%%
 We will prove that $h_U(X)$ is submodular in the following Lemma~\ref{lem:partMinSubmo}. 
 Since the sum of submodular functions is again submodular (see e.g.,~\cite{Fujishige}), 
 and we obtain the claim. 
\end{proof}

\begin{lemma} \label{lem:partMinSubmo}
 The set function $h_U$ defined by \eqref{def:phi} is submodular. 
\end{lemma}
\begin{proof}
 For convenience, 
 let 
  $X' = X \cap U$ and $Y' = Y \cap U$, 
  where we may assume that 
  $|X'| = |X \cap U| \leq |Y'| = |Y \cap U|$ holds, without loss of generality.
 Then, 
\begin{eqnarray}
 h_U(X) 
= \begin{cases} 
  |X \cap U| = |X'| &  \mbox{if $|X'| \leq |U|/2$,}\\
  |U \setminus X| = |U \setminus X'| = |U| - |X'| & \mbox{otherwise,} 
\end{cases}
\end{eqnarray}
 hold for any $X \in 2^V$. 
%%%%%%%
 We consider the following three cases. 

Case i) 
 Suppose that $|X'| \le |U|/2$ and $|Y'| \le |U|/2$ hold. 
 Then, $h_U(X) = |X'|$ and $h_U(Y) = |Y'|$ hold. 
  Since $|X' \cap Y'| \leq |X'| \leq |U|/2$, 
\begin{align*}
  h_U(X\cap Y) 
  = \min\{ |(X \cap Y) \cap U|, |U \setminus (X \cap Y)| \} 
  = |X' \cap Y'| 
\end{align*}
  hold. 
Observe that 
\begin{align*}
  h_U(X\cup Y) 
  &= \min\{ |(X \cup Y) \cap U|, |U \setminus (X \cup Y)| \}
  \leq |(X \cup Y) \cap U| = |X' \cup Y'| 
\end{align*}
 always hold. 
 Thus, 
\begin{align*}
\Phi_{h_U} (X,Y)
 &= h_U(X) + h_U(Y) - h_U(X\cup Y) - h_U(X\cap Y) \\
 & \geq |X'|+|Y'| -|X' \cup Y'| - |X' \cap Y'| =0 
\end{align*}
 hold where the last equality follows from that the cardinality function is modular. 
 We obtain the claim in the case. 

Case ii) 
 Suppose that 
  $|X'| > |S|/2$ and $|Y'| > |S|/2$ hold.  
 Then, $h_U(X) = |U|-|X'|$ and $h_U(Y) = |U|-|Y'|$ hold. 
 Since $|X' \cup Y'| \geq |Y'| > |U|/2$, 
\begin{align*}
  h_U(X\cup Y) 
  = \min\{ |(X \cup Y) \cap U|, |U \setminus (X \cup Y)| \} 
  = |U \setminus (X \cup Y)|
  = |U| - |X' \cup Y'| 
\end{align*}
  hold. 
Observe that 
\begin{align*}
  h_U(X\cap Y) 
  &= \min\{ |(X \cap Y) \cap U|, |U \setminus (X \cap Y)| \}
  \leq |U \setminus (X \cap Y)| = |U |- |X' \cap Y'| 
\end{align*}
 always holds. 
 Thus, 
\begin{align*}
\Phi_{h_U} (X,Y)
 &= h_U(X) + h_U(Y) - h_U(X\cup Y) - h_U(X\cap Y) \\
 & \geq (|U|-|X'|)+(|U|-|Y'|) - (|U|-|X' \cup Y'|) - (|U|-|X' \cap Y'|) =0 
\end{align*}
 hold where the last equality follows that the cardinality function is modular. 
 We obtain the claim in the case.

Case iii) 
 Suppose that 
  $|X'| \le |S|/2$  and $|Y'| > |S|/2$ hold. 
 Then, $h_U(X) = |X'|$ and $h_U(Y) = |U|-|Y'|$ hold. 
 Since $|X' \cup Y'| \geq |Y'| > |U|/2$, 
\begin{align*}
  h_U(X\cup Y) 
  = \min\{ |(X \cup Y) \cap U|, |U \setminus (X \cup Y)| \} 
  = |U| - |X' \cup Y'| 
\end{align*}
  holds. 
 Similarly, since $|X' \cap Y'| \leq |Y'| \leq |U|/2$, 
\begin{align*}
  h_U(X\cap Y) 
  = \min\{ |(X \cap Y) \cap U|, |U \setminus (X \cap Y)| \} 
  = |X' \cap Y'| 
\end{align*}
  holds. 
 Thus, 
\begin{align*}
\Phi_{h_U} (X,Y)
 &= h_U(X) + h_U(Y) - h_U(X\cup Y) - h_U(X\cap Y) \\
 & = |X'|+(|U|-|Y'|) - (|U|-|X' \cup Y'|) - (|X' \cap Y'|) \\
 & = 2|U| + (|X'|-|X' \cap Y'|) + (|X' \cup Y'| -|Y'|) \\
 &\geq 0
\end{align*}
 holds. 
 We obtain the claim.
\end{proof}

\subsection{Proof of Proposition~\ref{prop:ex-part2}}
\begin{proposition}[Proposition~\ref{prop:ex-part2}]
 For the submodular function $f$ given by \eqref{ex:part}, 
 $f \circp \sigma_S$ is submodular if and only if $S \in \mathcal{U}$. 
\end{proposition}

\begin{proof}
($\Leftarrow$)
 We show that $S  \in \mathcal{U}$ is a canonical set. 
 Let $g=f \circp \sigma_S$. 
 Then 
\begin{align}
 g(X) 
  = f \circp \sigma_S (X) = f(X \symmdiff S)
  = \min_{W \in {\cal U}} |(X \symmdiff S) \symmdiff W| 
  = \min_{W \in {\cal U}} |X \symmdiff (S \symmdiff W)| 
\label{eq:171106a}
\end{align}
   holds for any $X \in 2^V$.  
 Notice that 
    $W' = W \symmdiff S$ is in ${\cal U}$ 
    for any $W \in {\cal U}$ and  $S \in {\cal U}$. 
 Thus, 
\begin{align*}
\eqref{eq:171106a} 
   = \min_{W' \in {\cal U}} |X \symmdiff W'|, 
\end{align*}
 which implies that $g = f$, and hence $g$ is subdmodular by Proposition~\ref{prop:ex-part1}. 

($\Rightarrow$)
 We prove the contraposition: 
  if $S \notin {\cal U}$ then $g = f \circp \sigma_S$ is not submodular. 
 By the hypothesis that $S \notin\mathcal{U}$, 
 there exists $U_i$ such that 
  $S \cap U_i \neq \emptyset$ and $S \cap U_i \neq U_i$. 
 Let $X = S \symmdiff U_i$, and 
  we claim that $\Phi_g(X,S)<0$. 
 First, remark that 
  $g(X) = g(U_i \symmdiff S) = f (U_i) = 0$ and 
  $g(S) =  f (\emptyset) = 0$ hold. 
 Next, 
\begin{align*}
  g(X \cup S) &= g(S \cup U_i) = f ((S \cup U_i) \symmdiff S)= f (U_i \setminus S) > 0
\end{align*}
  where the last inequality follows from 
  the assumption $U_i \cap T \neq U_i$ and 
  the fact that $f(X)>0$ unless $X \in {\cal U}$ by the definition of $f$ (recall  \eqref{ex:part})). 
 Similarly,  
\begin{align*}
    g(X \cap S)
    &= g((S \symmdiff U_i) \cap T) = g(S \setminus U_i) \\
    &= f ((S \setminus U_i) \symmdiff S) 
    = f (S \setminus (S \setminus U_i))
    = f (S \cap U_i) \\
    &>0
\end{align*}
 hold where the last inequality follows from the assumption $S \cap U_i \neq \emptyset$. 
Thus, 
\begin{align*}
 \Phi_g(X,S)
 &= g(X) + g(S) - g(X \cup S) - g(X \cap S) \\
 &< 0 
\end{align*}
 hold, and we obtain the claim. 
\end{proof}

%%%%%%%%%%%%%%%%%%%%%%%%%%%%%%%%%%%%%%%%%%%%%%%
\section{Supplement to Section~\ref{sec:inseparable}}\label{apx:inseparable}
%%%%%%%%%%%%%%%%%%%
 This section shows the connection between 
   the connected components of the inequality graph $G_f$ given in Section~\ref{sec:char2} and  
   the inseparable decomposition (cf.\ \cite{BCT85,Queyranne98,Bach,Kamiyama,Fujishige,Fujishige17}) for submodular functions. 
 Precisely, we show the following. 
\begin{proposition}\label{prop:separable}
 Let $f \colon 2^V \to \mathbb{R}$ be a submodular function. 
 For any set $U \subseteq V$, 
  $\Phi_f(U,\overline{U})=0$ holds if and only if 
  $U$ and $\overline{U}$ are disconnected in the inequality graph $G_f$, % of $f$, 
  where $\overline{U} = V \setminus U$. 
\end{proposition}
 Notice that 
  Proposition~\ref{prop:separable} implies that 
  $U_1,\ldots,U_k$ are inseparable decomposition of $f$ 
  if and only if each $U_i$ is a connected component of $G_f$. 
%%%%%%%%%%%%%
 To prove Proposition~\ref{prop:separable}, 
   we will use the following Corollary~\ref{cor:Bach1} of Theorem~\ref{thm:Bach1o}. 
 In fact, the following Theorem~\ref{thm:Bach1o} is a part of Theorem~\ref{thm:inseparable}. 
 Here we will give a simpler proof 
  in a naive way without using the arguments on a base polytope. 
%%%%%%%%%%%%%%%%%%%%%%%%%%%%%%%%%%%%%%%%%%%%%%%%%%%%%%%%%%%%%%%%%%
\begin{theorem}[cf.~\cite{BCT85,Queyranne98,Bach,Kamiyama,Fujishige,Fujishige17}]\label{thm:Bach1o}
 Let $\rho \colon 2^V \to \mathbb{R}$ be a subdmodular function satisfying that $\rho(\emptyset) = 0$. 
 Suppose for $U \subset V$ ($U \neq \emptyset$) 
  that $\rho(V) = \rho(U) + \rho(\overline{U})$ holds where $\overline{U}=V \setminus U$. 
 Then, 
\begin{eqnarray}
 \rho(X) = \rho(X \cap U) + \rho(X \cap \overline{U})
\label{eq:171111f}
\end{eqnarray}
  holds for any $X \in 2^V$. 
\end{theorem}
\begin{proof}
 To begin with, we remark that \eqref{eq:171111f} is trivial for $X$ satisfying $X \subseteq U$ or $X \subseteq \overline{U}$. 
 Thus, we prove \eqref{eq:171111f} for $X$ 
   satisfying both $X \cap U \neq \emptyset$ and $X \cap \overline{U} \neq \emptyset$. 
 Since $\rho$ is submodular and $\rho(\emptyset) = 0$,
\begin{align}
 \rho(X) & \leq \rho(X \cap U) + \rho(X \cap \overline{U} ) \label{eq:171111a}\\
 \rho(X \cup U)+\rho(X \cap U) &\leq \rho(X) + \rho(U)  \label{eq:171111b}\\
 \rho(X \cup \overline{U})+\rho(X \cap \overline{U}) &\leq \rho(X) + \rho(\overline{U}) \label{eq:171111c}\\
 \rho(V)+\rho(X) &\leq \rho(X \cup U) + \rho(X \cup \overline{U}) \label{eq:171111d}
\end{align}
  hold, respectively. 
 By summing up \eqref{eq:171111b}, \eqref{eq:171111c} and \eqref{eq:171111d}, 
  we obtain that 
\begin{align}
 \rho(X \cap U) + \rho(X \cap \overline{U}) &\leq  \rho(X) \label{eq:171111e}
\end{align}
 holds, where we used the hypothesis that $\rho(V) = \rho(U) + \rho(\overline{U})$. 
 Now, \eqref{eq:171111a} and \eqref{eq:171111e} imply \eqref{eq:171111f}. 
\end{proof}

 As a corollary of Theorem~\ref{thm:Bach1o}, we obtain the following. 
%%%%%%%%%%%%%%%%
\begin{corollary}\label{cor:Bach1}
 Let $f \colon 2^V \to \mathbb{R}$ be a subdmodular function (and $f(\emptyset) = 0$ may not hold). 
 Suppose for $U \subset V$ ($U \neq \emptyset$) 
 that $f(V) + f(\emptyset) = f(U) + f(\overline{U})$ holds where $\overline{U}=V \setminus U$. 
 Then, 
\begin{eqnarray}
 f(X)+ f(\emptyset) = f(X \cap U) + f(X \cap \overline{U})
\label{eq:171111g}
\end{eqnarray}
  holds for any $X \in 2^V$. 
\end{corollary}
\begin{proof}
 Let $\rho(X) = f(X) - f(\emptyset)$ for any $X \in 2^V$, 
  then $\rho$ satisfies the hypothesis of Theorem~\ref{thm:Bach1o}. 
 Notice that $f(X) = \rho(X) + f(\emptyset)$ holds for any $X \in 2^V$. 
 Thus, \eqref{eq:171111f} implies \eqref{eq:171111g}. 
\end{proof}

%%%%%%%%%%%%%%%%%%%%%%%%%%%
Now we prove  Proposition~\ref{prop:separable}. 
\begin{proof}[Proof of Proposition~\ref{prop:separable}]
($\Leftarrow$)
  Suppose that $U$ and $\overline{U}$ are disconnected in $G_f$. 
 Then, we prove
\begin{eqnarray}
 f(X)+ f(\emptyset) = f(X \cap U) + f(X \cap \overline{U})
\label{eq:171118a}
\end{eqnarray}
  holds for any $X \in 2^V$, by an induction of the size $|X|$. 
 Notice that \eqref{eq:171118a} is trivial for $|X|=0$ and $|X|=1$. 
%%%%%
 Inductively assuming that \eqref{eq:171118a} holds for any $Y \in 2^V$ satisfying $|Y| \leq k$, 
  we prove \eqref{eq:171118a} for any $X \in 2^V$ satisfying $|X| = k+1$. 
 Notice that \eqref{eq:171118a} is trivial if $X \subseteq U$ or $X \subseteq \overline{U}$, 
 and we assume that $X \cap U \neq \emptyset$ and $X \cap \overline{U} \neq \emptyset$.
 Let $u \in X \cap U$, $v \in X \cap \overline{U}$, and $X'=X \setminus \{u,v\}$. 
 Since $\{u,v\} \in E_f$, 
\begin{eqnarray}
 f(X' \cup \{u\} \cup \{v\}) + f(X') &=&  f(X' \cup \{u\}) + f(X' \cup \{v\})
\label{eq:171118b}
\end{eqnarray}
 holds. 
 By the induction hypothesis, we obtain that   
%%%%%%%%
\begin{eqnarray}
 f(X' \cup \{u\})+ f(\emptyset) &=& f((X' \cup\{u\}) \cap U) + f(X' \cap \overline{U}) \label{eq:171118c}\\
 f(X' \cup \{v\})+ f(\emptyset) &=& f(X' \cap U) + f((X' \cup \{v\}) \cap \overline{U})\label{eq:171118d}
\end{eqnarray}
 hold, as well as 
\begin{eqnarray}
 f(X' \cap U) + f(X' \cap \overline{U}) &=&  f(X')+ f(\emptyset) 
\label{eq:171118e}
\end{eqnarray}
 holds. 
 By summing up \eqref{eq:171118b}--\eqref{eq:171118e}, we obtain 
\begin{eqnarray*}
 f((X' \cup \{u\}) \cup \{v\}) + f(\emptyset) &=&  f((X' \cup \{u\}) \cap U) + f((X' \cup \{v\}) \cap \overline{U}), 
\end{eqnarray*}
 which implies \eqref{eq:171118a} holds for the $X$. 
 Let $X=V$, then we obtain the claim.

%%%%%%%%%%%%%%%%%%%%%%%%%%
\begin{figure}[tbp]
\begin{center}
 \includegraphics[width=4cm,pagebox=cropbox,clip]{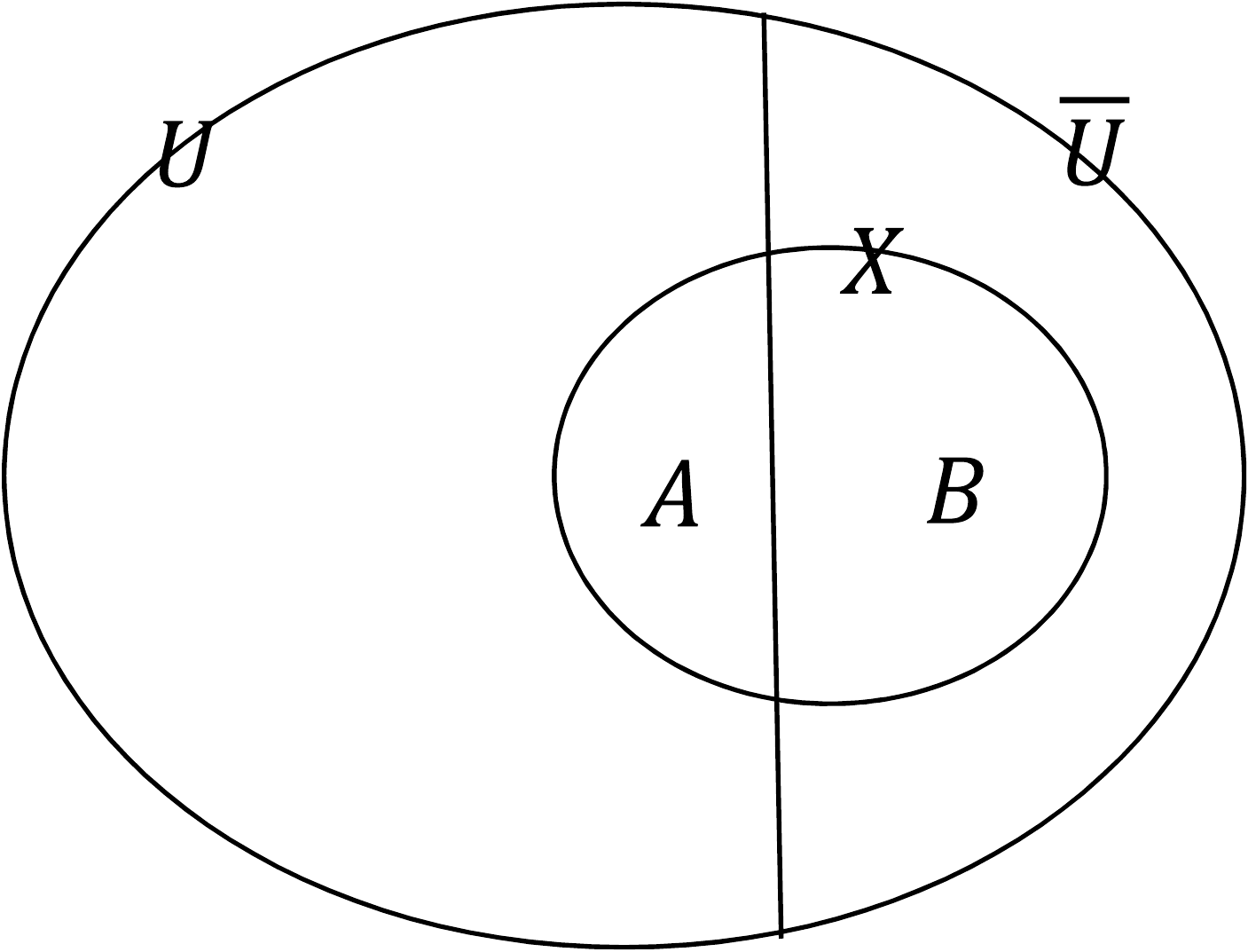}
\end{center}
 \caption{For Proof of Proposition~\ref{prop:separable} ($\Rightarrow$) }\label{fig:sep1}
\end{figure}
%%%%%%%%%%%%%%%%%%%%%
($\Rightarrow$)
Suppose 
 $f(U) + f(\overline{U}) = f(V) + f(\emptyset)$ holds. 
%%%%%%%%%%%%%%
 Assume for a contradiction that there is a pair $u \in U$ and $v \in \overline{U}$ such that $\{u,v\} \in E_f$. 
 This means that there is $X \subseteq V \setminus \{u,v\}$ 
 such that 
  $\check{\Phi}(X,\{u,v\}) \neq 0$ holds, by the definition~\eqref{def:Ef} of $E_f$.  
 For convenience, let $A = X \cap U$ and let $B =X \cap \overline{U}$ (see Figure~\ref{fig:sep1}). 
 By Corollary~\ref{cor:Bach1}, 
\begin{eqnarray}
 f(X \cup \{u\}) + f(\emptyset) &=& f(A \cup \{u\}) + f(B) \label{tmp171117a}\\
 f(X \cup \{v\}) + f(\emptyset)  &=& f(A) + f(B \cup \{v\}) \label{tmp171117b}\\
 f(X \cup \{u,v\})  + f(\emptyset) &=& f(A \cup \{u\}) + f(B \cup \{v\}) \label{tmp171117c}\\
 f(X)  + f(\emptyset) &=& f(A) + f(B) \label{tmp171117d}
\end{eqnarray} 
 hold, respectively. 
 Thus, 
\begin{align*}
 \check{\Phi}_f(X,\{u,v\}) 
 &= f(X \cup \{u\})+f(X \cup \{v\})-f(X \cup \{u,v\}) -f(X) \\
 &= f(A \cup \{u\}) + f(B) + f(A) + f(B \cup \{v\})\\
 &\qquad
 - (f(A \cup \{u\}) + f(B \cup \{v\})) 
 - (f(A) + f(B)) 
 &&(\mbox{by \eqref{tmp171117a}--\eqref{tmp171117d}})\\
 &= 0
\end{align*}
 holds, which contradicts to the assumption that $(X,\{u,v\}) \in {\cal P}$ satisfies $\check{\Phi}_f(X,\{u,v\}) \neq 0$. 
\end{proof}
\end{document}